\documentclass{article}
\usepackage[margin=1in]{geometry}

\usepackage[utf8]{inputenc}
\usepackage[T1]{fontenc}
\usepackage{xcolor}
\usepackage{graphicx}
\usepackage{svg}
\usepackage{subcaption}
\usepackage[font=small,labelfont=bf]{caption}
\usepackage{wrapfig}

\usepackage{tikz,pgfplots}
\usepackage{comment}
\usetikzlibrary{patterns}
\usetikzlibrary{backgrounds}
\usetikzlibrary{scopes}
\usetikzlibrary{arrows, automata, positioning}
\usetikzlibrary{decorations.markings}
\usepackage{verbatim}
\usepackage{multicol}

\usepackage{amsmath,amssymb,amsthm,amsfonts}
\usepackage{mathtools}
\usepackage{bbm}

\usepackage{microtype}
\usepackage{nicefrac}
\usepackage{booktabs}
\usepackage{enumitem}

\usepackage{hyperref}
\hypersetup{
    colorlinks=true,
    linkcolor=blue,
    filecolor=magenta,
    urlcolor=red,
    citecolor=red
}
\usepackage[capitalise,nameinlink]{cleveref}

\usepackage{thmtools}

\usepackage{tcolorbox}
\usepackage{framed}
\usepackage{tikz}
\usepackage{xspace}
\usepackage{cite}
\usepackage{authblk}
\usepackage[ruled,vlined]{algorithm2e}
\usepackage{setspace}

\newtheorem{theorem}{Theorem}
\newtheorem{lemma}[theorem]{Lemma}
\newtheorem{proposition}[theorem]{Proposition}
\newtheorem{corollary}[theorem]{Corollary}

\newtheorem*{framework*}{Framework}

\title{On the 2D Demand Bin Packing Problem: Hardness and Approximation Algorithms}

\author[1]{Susanne Albers}
\author[2]{Waldo Gálvez}
\author[3]{Ömer Behic Özdemir}

\affil[2]{Universidad de Concepción, Chile, \texttt{wgalvez@inf.udec.cl}}
\affil[1]{Technical University of Munich, Germany, \texttt{albers@in.tum.de}}
\affil[3]{Technical University of Munich, Germany, \texttt{omerbehicozdemir@gmail.com}}

\date{}

\begin{document}

\clearpage
\pagenumbering{arabic}

\maketitle
\begin{abstract}
    We study a two-dimensional generalization of the classical Bin Packing problem, denoted as 2D Demand Bin Packing. In this context, each bin is a horizontal timeline, and rectangular tasks (representing electric appliances or computational requirements) must be allocated into the minimum number of bins so that the sum of the heights of tasks at any point in time is at most a given constant capacity. We prove that simple variants of the problem are NP-hard to approximate within a factor better than $2$, namely when tasks have short height and when they are squares, and provide best-possible approximation algorithms for them; we also present a simple $3$-approximation for the general case. All our algorithms are based on a general framework that computes structured solutions for relatively large tasks, while including relatively small tasks on top via a generalization of the well-known First-Fit algorithm for Bin Packing.
\end{abstract}

\section{Introduction}\label{intro}

In this article, we consider a two-dimensional generalization of the classical Bin Packing problem, denoted as \emph{2D Demand Bin Packing}. An instance of the problem is defined by bins, where each bin corresponds to an axis-parallel rectangular region in the plane of height $C$ and width $T$, and a set of $n$ tasks, where each task $i$ is a rectangle characterized by its height $h_i$ and width $w_i$. The goal is to partition the given tasks into the minimum number of groups such that each group can be feasibly allocated into a bin; in this context, an allocation of a subset $I$ of tasks is a function that specifies, for each task $i$, the coordinate in the $X$-axis where the task starts (i.e., if task $i$ starts at $t\le T-w_i$, it is placed in the interval $(t,t+w_i)$). An allocation into a bin is feasible if, for each $t \in [0,T]$, the total height of tasks intersecting $t$ is at most $C$. See Figure~\ref{fig:gapinstance} for an example.

A closely related problem is the well-studied \emph{2D Geometric Bin Packing} problem~\cite{2d,BK14}, where the input is the same as before but the tasks are required to be embedded as non-overlapping axis-parallel rectangles into the bins\footnote{2D Demand Bin Packing can be thought as a version of 2D Geometric Bin Packing where each rectangle $i$ is replaced by $w_i$ rectangles of height $h_i$ and width $1$, but these slices must be placed into consecutive horizontal positions in the same bin}. The 2D Geometric Bin Packing problem models scenarios where tasks require contiguous and fixed portions of a given resource, such as allocating consecutive frequency or memory locations to given requirements, or cutting pieces out of sheets of some raw material (such as leather, wood, metal, etc.). However, Demand Allocation problems such as 2D Demand Bin Packing have received increasing attention because they better capture scenarios where geometric constraints are not required. For example, in the context of Energy Management, we might think of tasks as electric appliances, whose height corresponds to the amount of energy they consume, and their width corresponds to the amount of time they are required; hence, 2D Demand Bin Packing would correspond to deciding how many periods of $T$ units of time are required to execute all the tasks without overloading the system capacity $C$ (see~\cite{THLW13,ABCGJKLP13} for other examples of similar scenarios in the context of Smart-Grid Scheduling). It is worth remarking that geometric constraints do not play a role here as we only care about the accumulated energy consumption at each point in time; furthermore, enforcing such constraints might lead to inefficient solutions (see for instance Figure~\ref{fig:gapinstance}), suggesting that techniques specifically designed for the Demand Allocation context are required.

It is not difficult to verify that it is NP-hard to decide if a set of given tasks can be allocated into one bin or not (for example via a reduction from the \emph{Partition} problem, encoding the Partition instance into the widths of the tasks and the heights being half the height of the bin). This already implies that, for any $\varepsilon>0$, there is no $(2-\varepsilon)$-approximation for 2D Demand Bin Packing (and also for 2D Geometric Bin Packing). It is known that 2D Geometric Bin Packing admits a $2$-approximation which is best possible unless P=NP~\cite{2d}; however, to the best of our knowledge, 2D Demand Bin Packing has not been explicitly studied before.

	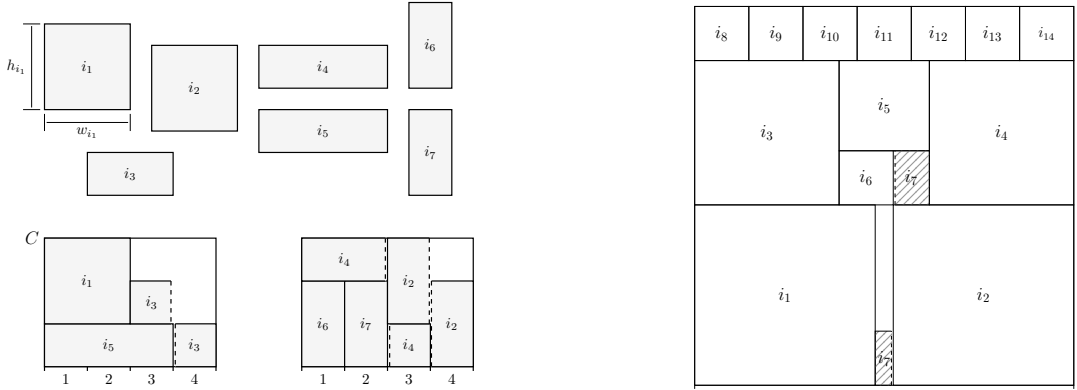
\begin{figure}
    
	    \centering
    
		\scalebox{.57}{\begin{tikzpicture}
			   \draw[thick] (0, 0) -- (0,-0.1);
		   %\draw[thick] (0.5, 0) -- (0.5,-0.1);
		   \draw[thick] (1, 0) -- (1,-0.1);
		   %\draw[thick] (1.5, 0) -- (1.5,-0.1);
		   \draw[thick] (2, 0) -- (2,-0.1);
		   %\draw[thick] (2.5, 0) -- (2.5,-0.1);
		   \draw[thick] (3, 0) -- (3,-0.1);
		   %\draw[thick] (3.5, 0) -- (3.5,-0.1);
		   \draw[thick] (4, 0) -- (4,-0.1);
		   %\draw[thick] (4.5, 0) -- (4.5,-0.1);
		   %\draw[thick] (5, 0) -- (5,-0.1);
		   %\draw[thick] (5.5, 0) -- (5.5,-0.1);
		   \draw[thick] (6, 0) -- (6,-0.1);
		   %\draw[thick] (6.5, 0) -- (6.5,-0.1);
		   \draw[thick] (7, 0) -- (7,-0.1);
		   %\draw[thick] (7.5, 0) -- (7.5,-0.1);
		   \draw[thick] (8, 0) -- (8,-0.1);
		   %\draw[thick] (8.5, 0) -- (8.5,-0.1);
		   \draw[thick] (9, 0) -- (9,-0.1);
		   %\draw[thick] (9.5, 0) -- (9.5,-0.1);
		   \draw[thick] (10, 0) -- (10,-0.1);
		   %\draw[thick] (10.5, 0) -- (10.5,-0.1);
		   
            %\draw (0, -0.3) node {$0 $};
            \draw (0.5, -0.3) node {$1$};
            %\draw (1, -0.3) node {$2$};
            \draw (1.5, -0.3) node {$2$};
            %\draw (2, -0.3) node {$4 $};
            \draw (2.5, -0.3) node {$3$};
            %\draw (3, -0.3) node {$6$};
            \draw (3.5, -0.3) node {$4$};
            %\draw (4, -0.3) node {$8 $};
            %\draw (4.5, -0.3) node {$9$};
            %\draw (5, -0.3) node {$10$};
            %\draw (5.5, -0.3) node {$11 $};	
            %\draw (6, -0.3) node {$12 $};
            \draw (6.5, -0.3) node {$1$};
            %\draw (7, -0.3) node {$14$};
            \draw (7.5, -0.3) node {$2$};
            %\draw (8, -0.3) node {$16 $};
            \draw (8.5, -0.3) node {$3$};
            %\draw (9, -0.3) node {$18$};
            \draw (9.5, -0.3) node {$4$};
           %\draw (10, -0.3) node {$20 $};
            %\draw (10.5, -0.3) node {$21$};        
            
		  \draw[thick] (0,0) rectangle (4,3);
          \draw[thick] (6,0) rectangle (10,3);

            \draw[thick, fill=lightgray!15] (0,6) rectangle (2,8);
            \draw[thick, fill=lightgray!15] (2.5,5.5) rectangle (4.5,7.5);
            \draw[thick, fill=lightgray!15] (5,7.5) rectangle (8,6.5);
            \draw[thick, fill=lightgray!15] (5,6) rectangle (8,5);
            \draw[thick, fill=lightgray!15] (1,4) rectangle (3,5);
            \draw[thick, fill=lightgray!15] (8.5,4) rectangle (9.5,6);
            \draw[thick, fill=lightgray!15] (8.5,6.5) rectangle (9.5,8.5);
  		 
          \draw (-0.1, 8) -- (-0.5,8);
  		  \draw (-0.3, 7.95) -- (-0.3,6.05);
  		  \draw (-0.1, 6) -- (-0.5,6);
  		   \draw (-0.3, 7) node[anchor=east] {$h_{i_1}$};
  		   
          \draw (0, 5.9) -- (0,5.5);
  		  \draw (0.05, 5.7) -- (1.95,5.7);
  		  \draw (2, 5.9) -- (2,5.5);
  		   \draw (1, 5.7) node[anchor=north] {$w_{i_1}$};

           \draw (1,7) node {$i_1$};
           \draw (3.5,6.5) node {$i_2$};
           \draw (2,4.5) node {$i_3$};
           \draw (6.5,7) node {$i_4$};
           \draw (6.5,5.5) node {$i_5$};
           \draw (9,7.5) node {$i_6$};
           \draw (9,5) node {$i_7$};

           \draw (0,3) node[anchor=east] {\large $C$};

           \draw[thick, fill=lightgray!15] (0,0) rectangle (3,1);
           \draw[thick, fill=lightgray!15] (0,1) rectangle (2,3);
           \fill[color=lightgray!15] (2,1) rectangle (2.95,2);
           \fill[color=lightgray!15] (3.05,0) rectangle (4,1);
           \draw[thick, dashed] (2.95,2) -- (2.95,1);
           \draw[thick, dashed] (3.05,1) -- (3.05,0);
           \draw[thick] (2.95,2) -- (2,2) -- (2,1) -- (2.95,1);
           \draw[thick] (3.05,1) -- (4,1) -- (4,0) -- (3.05,0);

           \draw (1,2) node {$i_1$};
           \draw (1.5,0.5) node {$i_5$};
           \draw (2.5,1.5) node {$i_3$};
           \draw (3.5,0.5) node {$i_3$};

           \draw[thick, fill=lightgray!15] (6,0) rectangle (7,2);
           \draw[thick, fill=lightgray!15] (7,0) rectangle (8,2);
           \fill[color=lightgray!15] (6,2) rectangle (7.95,3);
           \fill[color=lightgray!15] (8.05,0) rectangle (9,1);
           \draw[thick, dashed] (7.95,3) -- (7.95,2);
           \draw[thick, dashed] (8.05,0) -- (8.05,1);
           \draw[thick] (7.95,3) -- (6,3) -- (6,2) -- (7.95,2);
           \draw[thick] (8.05,0) -- (9,0) -- (9,1) -- (8.05,1);
           \fill[color=lightgray!15] (8,1) rectangle (8.975,3);
           \fill[color=lightgray!15] (9.025,0) rectangle (10,2);
           \draw[thick,dashed] (8.975,3) -- (8.975,1);
           \draw[thick,dashed] (9.025,2) -- (9.025,0);
           \draw[thick] (8.975,3) -- (8,3) -- (8,1) -- (8.975,1);
           \draw[thick] (9.025,2) -- (10,2) -- (10,0) -- (9.025,0);

           \draw (6.5,1) node {$i_6$};
           \draw (7.5,1) node {$i_7$};
           \draw (7,2.5) node {$i_4$};
           \draw (8.5,0.5) node {$i_4$};
           \draw (8.5,2) node {$i_2$};
           \draw (9.5,1) node {$i_2$};

		\end{tikzpicture}}\hspace{80pt}
        \scalebox{.48}{\begin{tikzpicture}
		   \draw[thick] (0, 0) -- (0,-0.1);
		   %\draw[thick] (0.5, 0) -- (0.5,-0.1);
		   %\draw[thick] (1, 0) -- (1,-0.1);
		   %\draw[thick] (1.5, 0) -- (1.5,-0.1);
		   %\draw[thick] (2, 0) -- (2,-0.1);
		   %\draw[thick] (2.5, 0) -- (2.5,-0.1);
		   %\draw[thick] (3, 0) -- (3,-0.1);
		   %\draw[thick] (3.5, 0) -- (3.5,-0.1);
		   %\draw[thick] (4, 0) -- (4,-0.1);
		   %\draw[thick] (4.5, 0) -- (4.5,-0.1);
		   %\draw[thick] (5, 0) -- (5,-0.1);
		   %\draw[thick] (5.5, 0) -- (5.5,-0.1);
		   %\draw[thick] (6, 0) -- (6,-0.1);
		   %\draw[thick] (6.5, 0) -- (6.5,-0.1);
		   %\draw[thick] (7, 0) -- (7,-0.1);
		   %\draw[thick] (7.5, 0) -- (7.5,-0.1);
		   %\draw[thick] (8, 0) -- (8,-0.1);
		   %\draw[thick] (8.5, 0) -- (8.5,-0.1);
		   %\draw[thick] (9, 0) -- (9,-0.1);
		   %\draw[thick] (9.5, 0) -- (9.5,-0.1);
		   %\draw[thick] (10, 0) -- (10,-0.1);
		   \draw[thick] (10.5, 0) -- (10.5,-0.1);
		   
            %\draw (0, -0.3) node {\large $0 $};
            %\draw (0.5, -0.3) node {$1$};
            %\draw (1, -0.3) node {$2$};
            %\draw (1.5, -0.3) node {$3 $};
            %\draw (2, -0.3) node {$4 $};
            %\draw (2.5, -0.3) node {$5$};
            %\draw (3, -0.3) node {$6$};
            %\draw (3.5, -0.3) node {$7 $};
            %\draw (4, -0.3) node {$8 $};
            %\draw (4.5, -0.3) node {$9$};
            %\draw (5, -0.3) node {$10$};
            %\draw (5.5, -0.3) node {$11 $};	
            %\draw (6, -0.3) node {$12 $};
            %\draw (6.5, -0.3) node {$13$};
            %\draw (7, -0.3) node {$14$};
            %\draw (7.5, -0.3) node {$15 $};
            %\draw (8, -0.3) node {$16 $};
            %\draw (8.5, -0.3) node {$17$};
            %\draw (9, -0.3) node {$18$};
            %\draw (9.5, -0.3) node {$19 $};
            %\draw (10, -0.3) node {$20 $};
            %\draw (10.5, -0.3) node {\Large $21$};        

		  \draw[thick] (0,0) rectangle (10.5,10.5);
  		 \draw (0,0) rectangle (5,5);
  		 \draw (5.5,0) rectangle (10.5,5);
  		 \draw (2.5, 2.5) node {\Large $i_1 $};
  		  \draw (8, 2.5) node {\Large $i_2 $};
  		 \draw (0,5) rectangle (4,9);
  		   \draw (6.5,5) rectangle (10.5, 9);
  		   \draw (0,9) rectangle (1.5, 10.5);
  		    \draw (8.5, 7) node {\Large $i_4 $};
  		     \draw (2, 7) node {\Large $i_3 $};
  		       		   \draw (1.5,9) rectangle (3, 10.5);
  		       		   \draw (3,9) rectangle (4.5, 10.5);
  		       		    \draw (4.5,9) rectangle (6, 10.5);
  		       		    \draw (6,9) rectangle (7.5, 10.5);
  		       		     \draw (7.5,9) rectangle (9, 10.5);
  		       		     \draw (9,9) rectangle (10.5, 10.5);
  		       		     \draw (0.75, 9.75) node {\Large $i_8 $};
  		       		     \draw (2.25, 9.75) node {\Large $i_9 $};
  		       		     \draw (3.75, 9.75) node {\Large $i_{10} $};
  		       		     \draw (5.25, 9.75) node {\Large $i_{11}$};
  		       		     \draw (6.75, 9.75) node {\Large $i_{12}$};
  		       		     \draw (8.25, 9.75) node {\Large $i_{13}$};
  	                    \draw (9.75, 9.75) node {\large $i_{14}$};
          \draw (4,6.5) rectangle (6.5,9);
          \draw (5.25, 7.75) node {\Large $i_{5}$};
          \draw (4,5) rectangle (5.5,6.5);
           \draw (4.75, 5.75) node {\Large $i_{6}$};
          \fill[gray!20, pattern=north east lines, pattern color = gray!80, thick] (5.55, 5) rectangle (6.5,6.5);
           \draw (6, 5.75) node {\Large $i_{7}$};
           \fill[gray!20, pattern=north east lines, pattern color = gray!80, thick] (5, 0) rectangle (5.45,1.5);
           \draw (5.25, 0.75) node {\Large $i_{7}$};
           \draw[dashed] (5.45,0) -- (5.45,1.5);
           \draw[dashed] (5.55,5) -- (5.55,6.5);
           \draw (5.45,0) -- (5,0) -- (5,1.5) -- (5.45,1.5);
           \draw (5.55,5) -- (6.5,5) -- (6.5,6.5) -- (5.5,6.5);

		\end{tikzpicture}}
		\caption{\textbf{(Left)}: Example of an instance of 2D Demand Bin Packing, with parameters $T=4$ and $C=3$, and a feasible solution defined by two bins. \textbf{(Right)}: Example of an instance of 2D Demand Bin Packing of squares into a square bin that fits into one bin, but that does not allow a geometric packing of the squares into the bin (see Appendix~\ref{app:gap_instance}).}
		 \label{fig:gapinstance}

	\end{figure} 

\subsection{Our Results}

We provide best-possible approximation algorithms for two relevant special cases of 2D Demand Bin Packing that preserve the aforementioned hardness of approximation (see Appendix~\ref{app:hardness}): When tasks have height at most $\frac{1}{9}C$, denoted as \emph{2D Demand Bin Packing for short tasks}, and when tasks have equal height and width, denoted as \emph{2D Demand Bin Packing for square tasks}. We also provide a simple $3$-approximation for the general case.

Our results are based on a common framework that consists of two steps: We classify the tasks according to their heights and widths into \emph{large} and \emph{small} (the precise definition will depend on the context). Our main technical contribution states that, as long as large tasks can be allocated into bins so that each bin either has enough total area or has a special structure (namely that the sequence of total loads for each point in time is monotonically non-increasing), and a certain relation between the parameters defining large and small tasks is satisfied, then small tasks can be placed on top of the solution without adding extra bins (see Lemma~\ref{lem:smalltasks}). The algorithm to place small tasks is an adaptation to the two-dimensional case of the classical First-Fit algorithm for Bin Packing, specially tailored for the 2D Demand Allocation context.

The previous framework allows to focus solely on how to place large tasks with the required structure, which in most cases can be achieved by means of classical Bin Packing or 2D Geometric Bin Packing routines. Our results show that the 2D Demand Allocation context indeed allows for simpler algorithms in comparison to its geometric counterparts~\cite{Z05,2d}.

\subsection{Related Results}

The (one-dimensional) Bin Packing problem is one of the most classical and well-studied problems in Combinatorial Optimization. It is known to be strongly NP-hard and, for any $\varepsilon>0$, there cannot be a $\left(\frac{3}{2}-\varepsilon\right)$-approximation unless P=NP~\cite{GJ79}; the algorithms \emph{First-Fit Decreasing} and \emph{Best-Fit Decreasing} reach an approximation ratio of $\frac{3}{2}$, which is best possible~\cite{levi}. The asymptotic setting has also received considerable attention, with well-known approximation schemes~\cite{KK82,VL81}. Towards the big open problem of whether there exists a polynomial-time algorithm that always uses at most $OPT+1$ bins, the current best algorithm due to Hoberg and Rothvoss uses $OPT + O(\log(OPT))$ bins~\cite{HR17}. We refer the readers to the surveys of Coffman et al.~\cite{survey} and Delorme et al.~\cite{DIM16} for a detailed description of algorithmic results for the problem.

Among the possible generalizations of Bin Packing to multiple dimensions, the ones that have received more attention are \emph{Vector Bin Packing} and \emph{Geometric Bin Packing}. In the first problem, both the bins and the items are characterized by $d$-dimensional vectors, while in the second problem the bins and items are d-dimensional hyperrectangles. In the two-dimensional case, while Vector Bin Packing has the same hardness of approximation as Bin Packing, Geometric Bin Packing cannot be approximated within a factor better than $2$ even when the items are squares~\cite{leung}. For both problems, almost tight approximation algorithms are known to exist~\cite{BEK16,2d}, and there have been recent progress on the asymptotic regime~\cite{CC09,BK14,KMS23}.

Another (less straightforward) generalization of Bin Packing to the two-dimensional context is \emph{Strip Packing}. In the Geometric\footnote{In the literature, this problem usually does not include the word ``geometric'', but recent results include it to make a difference with the Demand Allocation version.} Strip Packing problem, we are given a semi-infinite axis-parallel strip of finite width and a set of rectangles, and the goal is to place the rectangles into the strip so that they are all axis-parallel, do not overlap, and the maximum height spanned by the rectangles is minimized. In the Demand Strip Packing problem, the input is the same but no geometric constraint is required, instead we only ask to minimize the maximum accumulated load along the $X$-axis in the strip. While both problems cannot be approximated with a factor better that $\frac{3}{2}$, the current best approximation ratio for Geometric Strip Packing is $\left(\frac{5}{3}+\varepsilon\right)$ due to Harren et al.~\cite{HJPS14}, but it was recently shown by Eberle et al. that Demand Strip Packing admits an almost tight $\left(\frac{3}{2}+\varepsilon\right)$-approximation~\cite{EHRW25}. This shows that the Demand Allocation context is an interesting intermediate step towards a better understanding of Geometric Packing problems. For a detailed description of other Geometric and multidimensional Packing problems, we refer the reader to the survey by Christensen et al.~\cite{CKPT17}.

\subsection{Organization of the paper}

In Section~\ref{sec:prelim}, we provide a formal definition of the problem with useful notation, and our general algorithmic approach. In Section~\ref{sec:specialalgo}, we show how to apply this approach to obtain best-possible algorithms for the case of short tasks and square tasks, and a simple $3$-approximation for the general case.

\section{Preliminaries}\label{sec:prelim}

In this section, we will provide a formal definition of the problem as well as useful notation. The definition we use slightly deviates from the description in the Introduction, but is more practical for our purposes. 

In the 2D Demand Bin Packing problem, we are given bins defined by a timeline that has $T\in \mathbb{N}$ time slots, and each time slot has a capacity $C\in \mathbb{N}$; we are also given a set $I$ of $n$ tasks, where each task $i$ is characterized by its height $h_i \in \mathbb{N}$ and width $w_i\in \mathbb{N}$. A feasible allocation of a subset of tasks into a bin corresponds to a function $\mathcal{A}:I \to \{1,\dots,T\}$ that specifies, for each task $i$,  the time slot where the task starts (i.e., if $\mathcal{A}(i) = t$, with $t\in \{1,\dots,T-w_i\}$, then load $h_i$ is placed into time slots $\{t,\dots,t+w_i-1\}$). An allocation is feasible if, for each time slot $t \in [0,T]$, the total load placed into $t$ is at most $C$ (see Figure~\ref{fig:gapinstance} for an example).

For a task $i$, we define the area of $i$ as $a(i) = h_iw_i$, which naturally extends to a subset $S$ of tasks as $a(S) = \sum_{i\in S}{a(i)}$; an analogous extension will be used for heights $h(S)$ and widths $w(S)$. Throughout this work, for a given instance $I$ of the problem, $OPT(I)$ will denote the optimal number of bins for $I$ (dependence on $I$ will be dropped if clear from the context). Notice that $OPT(I) \ge \lceil a(I) / (T\cdot C) \rceil$, as the total area of tasks that fit in a bin is at most $T\cdot C$. Given a feasible allocation $\mathcal{A}$ of a set of tasks into a bin, we define the \emph{load profile} of $\mathcal{A}$ as a vector of dimension $T$, where each entry $e\in \{1,\dots,T\}$ stores the total load in $e$, denoted as $\ell(e)$. The load profile of an allocation is \emph{sorted} if $\ell(1) \ge \ell(2) \ge \dots \geq \ell(T)$. If the set of tasks allocated into a bin has total area at least $\alpha T \cdot C$, $0\le \alpha\le 1$, we say that the allocation (or the bin) is \emph{$\alpha$-full}. 

\subsection{Our Algorithmic Approach}

%\textbf{Next-Fit Decreasing Height (NFDH):} The tasks are sorted non-increasingly. First task is placed at the bottom-left corner of the bin. Similarly to the first placed tasks of each level, this task is also determines the height of the above level. We try to fit each task into the current level. If it does not fit, we create a new level and pack the task there. The current level becomes this level. In each level, we pack from left to right. If the task does not fit into the bin, a new bin is opened and the task is placed at the bottom-left corner of the bin and this level becomes the current level.
  
One of the building blocks for our algorithms will be an adaptation of the classical \emph{First-Fit} algorithm for Bin Packing. Recall that, in this algorithm, each item is iteratively packed in the leftmost bin where it fits, opening a new bin if it does not fit in any of the currently open bins. We now present an algorithm based on this idea for our problem, which will be useful to allocate relatively small tasks over \emph{structured} solutions.

Suppose we are given a partial solution for a subset of the tasks such that each bin has a sorted load profile. We consider the remaining tasks in an arbitrary but fixed order and, for each possible time slot in the bins starting from left to right, the whole list of tasks is checked according to that order\footnote{The algorithm can be implemented in polynomial time by simply checking the time slots where the load profile changes instead of checking every time slot (see~\cite{GGJK21}).}. If the current task fits in the current position then we allocate it. If there are tasks which could not be placed in any position, the algorithm opens a new bin and starts checking that bin. This procedure will be referred to as the \emph{First-Fit inspired algorithm}. The key point of this algorithm is that, similarly to First-Fit, if we need to open a new bin then we have guarantees on the total area of tasks already placed in the bins, thus bounding their total number.

We will say that a partial solution for a subset of the tasks is $k$-structured, $k\ge 1$, if it uses at most $k\cdot OPT$ bins, and each bin is either $\frac{1}{k}$-full or has a sorted load profile. As the following lemma states, if $k$-structured solutions can be computed for relatively large tasks, then relatively small tasks can be included without requiring more than $k\cdot OPT$ bins in total by means of the First-Fit inspired algorithm, as long as a simple relation between the parameters is satisfied. 

\begin{lemma}\label{lem:smalltasks} 
  If there exists a $k$-structured solution, $k \ge 1$, for tasks with height $h_i > \delta_h C$ or width $w_i > \delta_w T$ from the instance, then it is possible to allocate all the tasks into at most $kOPT$ bins, if $(1-\delta_h)\cdot( 1-\delta_w) \ge \frac{1}{k}$.
 \end{lemma}
 
 \begin{proof}
Let $B \le kOPT$ be the number of bins from the initial $k$-structured solution whose allocations are $\frac{1}{k}$-full. We will use the First-Fit inspired algorithm to allocate the remaining tasks into the bins with sorted load profiles (possibly opening new bins in the process). Suppose by contradiction that we used more than $kOPT$ bins in total. This implies that tasks in the last bin have height at most $\delta_h C$ and width at most $\delta_w T$. 

We claim that all the bins except possibly for the last one have tasks allocated of total area at least $(C-\delta_h C)(T-\delta_w T)$. To see this, we observe that, up to the point where the last bin was open, the algorithm maintained the following invariant: in the bin currently being considered, the load profile starting from the current slot to the right is non-increasing. The invariant is initially satisfied as the bins have sorted load profiles, and whenever we allocate a task the invariant is maintained as the leftmost slots starting from the current one increase their load by the same amount while the remaining ones to the right are left untouched. Thanks to this invariant, and since the tasks assigned to the last bin were not assigned to the current time slot, every time we finished checking one of the leftmost $T-\delta_w T$ slots in the current bin its total load became larger than $C-\delta_h C$ (otherwise we would have packed any of the tasks in this bin). This implies that the total area of tasks allocated in each of these bins is at least $(1-\delta_h )(1-\delta_w )TC \ge \frac{TC}{k}$. Consequently, the total area placed in the bins so far is larger than $B\cdot\frac{1}{k}TC+(kOPT-B)\frac{TC}{k} \ge OPT \cdot TC$, which is a contradiction.
 \end{proof}

%\begin{lemma}\label{lem:state} 
%When applying the First-Fit inspired algorithm, all the bins except possibly for the last one have tasks allocated of total area at least $(H-h_{last})(W-w_{last})$, where $h_{last}$ (resp. $w_{last}$) is the height (resp. width) of any task assigned to the last bin.
% \end{lemma}

% \begin{proof} We will prove first that, up to the point where the last bin was open, the algorithm maintained the following invariant: in the bin currently being considered, the load profile starting from the current slot to the right is non-increasing. Notice that the invariant is satisfied initially as the bins have non-increasing load profiles, and whenever we allocate a task the invariant is maintained as the leftmost slots starting from the current increase by the same amount while the remaining ones to the right are left untouched. Thanks to this invariant, and since the tasks assigned to the last bin were not assigned to the current one, every time we finished checking one of the leftmost $W-w_{last}$ slots in the current bin its total load became larger than $H-h_{last}$ (otherwise we would have packed it in this bin). This implies that the total area of tasks allocated in each of these bins is at least $(H-h_{last})(W-w_{last})$. \end{proof}

 Thanks to this lemma, we can focus on computing solutions for large tasks with the required structure.

\section{Computing Structured Solutions in 2D Demand Bin Packing}\label{sec:specialalgo}

In this section, we will apply the previously described framework to obtain best-possible approximation algorithms for two special cases of 2D Demand Bin Packing, namely short tasks and square tasks, as well as a simple $3$-approximation for the general case. As mentioned before, we will  first search for structured solutions for \emph{large} tasks so as to pack the remaining \emph{small} tasks using Lemma~\ref{lem:smalltasks}, where the notion of large and small will be specified carefully in each case. We will assume along this section that $OPT$ is known to the algorithm, which holds without loss of generality as $OPT\le n$.

\subsection{2D Demand Bin Packing for short tasks}\label{sec:shorttasks}

We start by showing how to compute $2$-structured solutions for the case of short tasks, the case where $h_i \le \frac{1}{9}C$ for every $i\in\{1,\dots,n\}$. 
\begin{theorem}\label{lem:shortapx}
There is a 2-approximation for 2D Demand Bin Packing for short tasks.
\end{theorem}
\begin{proof}We will show first how to compute a $2$-structured solution for tasks having width $w_i > \frac{1}{3}T$. This way, the remaining tasks would have height at most $\frac{1}{9}C$ and width at most $w \le \frac{1}{3}T$, thus we obtain a solution that uses at most $2OPT$ bins thanks to Lemma~\ref{lem:smalltasks} since $\left( 1- \frac{1}{9} \right) \left( 1- \frac{1}{3} \right) = \frac{16}{27} > \frac{1}{2}$.

Consider a region of height $C\cdot OPT$ and width $T$. Sort the tasks non-increasingly by width, and let us label them in such a way that $w_1 \ge \dots \ge w_{n'}$. Let us place them one on top of each other in that order, starting at the left boundary of the region, until their total height becomes larger than $C\cdot OPT$ for the first time. We remove the last task assigned $k$, and place the remaining tasks in the same order on the right boundary of the region starting from the top (see Figure~\ref{fig:short2apxf1}). Notice that the total height of tasks having width larger than $\frac{1}{3}T$ is at most $2C\cdot OPT$, so the second group of tasks has total height at most $C\cdot OPT$. We will prove now that this geometric placement is feasible.
	
Suppose by contradiction that two tasks $i'$ and $i^*$, $i'\le i^*$, overlap due to their total width exceeding $1$. Let $I_1=\{1,\dots,i'-1\}$ be the set of tasks below $w_{i'}$, and let $I_2$ be the set of tasks $\{(i' + 1),...,i^*\}$, which have width at least $w_{i^*}$ and total height larger than $2(OPT-h(I_1))$ as task $k$ is also included. Observe that the total height of tasks having width larger than $\frac{T}{2}$ is at most $OPT\cdot C$, and hence $w_{i^*} \le \frac{T}{2}$, implying that the tasks in $I_1$ have width larger than $T - w_{i^*} > \frac{T}{2}$. For any feasible solution to the 2D Demand Bin Packing problem we have that, in each bin, if the total height of tasks from $I_1$ is $p$, then the total height of tasks from $I_2$ in the bin is at most $2(1-p)$ as tasks from $I_1$ and $I_2$ always overlap and their width is larger than $\frac{1}{3}$. This means that the total height of tasks in $I_2$ is at most $2(OPT-h(I_1))$, which is a contradiction. 

We will now derive a packing into $OPT$ bins from the previous construction. Let us initially simply draw horizontal lines at each possible integral height and remove the tasks that are intersected (notice that this is well defined as in the previous packing no task is vertically sliced). This induces a solution for the non-intersected tasks. Furthermore, tasks intersected by $9$ of the boundaries can be allocated into one bin. Hence $\lceil \frac{1}{9}(OPT-1) \rceil$ bins are used to pack these intersected tasks. This induces a packing of $OPT + \lceil \frac{1}{9}(OPT-1) \rceil + 1$ in total ($1$ more for the dropped task $k$). Notice that the bins with intersected tasks are packed as 'shelves', hence being well-structured as their load profiles are non-increasing. Bins having only tasks on the left side have non-increasing profiles too. The remaining bins have tasks both to the left and to the right, but notice that all of them except for one (the topmost ones in the construction) have tasks inside of total area at least $2 \left( 1- \frac{2}{9} \right) \frac{1}{3} = \frac{14}{27} > \frac{1}{2}$ as both the left and right pile completely cover these bins before removing the intersected tasks. It only remains to consider the bin which is partially filled in the right and the dropped task, but in this bin if the right pile inside has total height at least $\frac{7}{9}$ the area of tasks in the bin is at least $\frac{1}{2}$ analogously as before, otherwise they can be placed on top of task $k$, obtaining a non-increasing profile. 

In conclusion, we obtain a feasible well-structured solution using at most $OPT + \lceil \frac{1}{9}(OPT-1) \rceil + 1$. If $OPT=1$, this value is at most $OPT + 1 = 2OPT$. If $OPT \ge 2$, then the used number of bins is at most
\[OPT + \left \lceil \frac{1}{9}(OPT-1) \right \rceil + 1 \le OPT + \frac{1}{9}(OPT-1) + \frac{8}{9} + 1 = \frac{10}{9}OPT + \frac{16}{9} \le 2OPT.\] \end{proof}

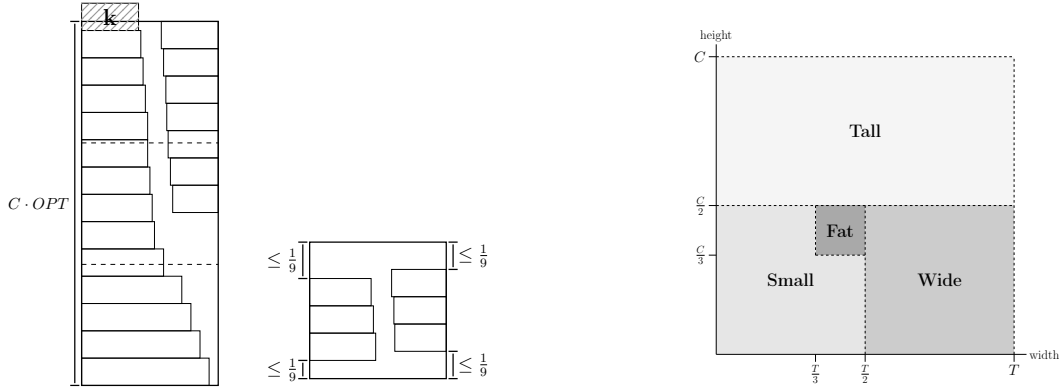
\begin{figure}
	    \centering
        
		\resizebox{0.4\textwidth}{!}{
		\begin{tikzpicture}
		  \draw[thick] (0,0) rectangle (3,8);
  	       \draw (0,0) rectangle (2.8,0.6);
  	       \draw (0,0.6) rectangle (2.6,1.2);
  	       \draw (0,1.2) rectangle (2.4,1.8);
  	       \draw (0,1.8) rectangle (2.2,2.4);
           \draw (0,2.4) rectangle (1.8,3);
         \draw (0,3) rectangle (1.6,3.6);
           \draw (0,3.6) rectangle (1.55,4.2);
           \draw (0,4.2) rectangle (1.5,4.8);
           \draw (0,4.8) rectangle (1.45,5.4);
           \draw (0,5.4) rectangle (1.45,6);
           \draw (0,6) rectangle (1.4,6.6);
           \draw (0,6.6) rectangle (1.35,7.2);
           \draw (0,7.2) rectangle (1.3,7.8);
           \draw (0,7.8) rectangle (1.25,8.4);

           \draw (1.75,7.4) rectangle (3,8);
           \draw (1.8,6.8) rectangle (3,7.4);
           \draw (1.85,6.2) rectangle (3,6.8);
           \draw (1.87,5.6) rectangle (3,6.2);
           \draw (1.9,5) rectangle (3,5.6);
           \draw (1.95,4.4) rectangle (3,5);
           \draw (2,3.8) rectangle (3,4.4);

  		  \draw[dashed] (0, 2.66) -- (3, 2.66);
  		  \draw[dashed] (0, 5.33) -- (3, 5.33);

  	 \fill[pattern = north east lines, pattern color = gray!80] (0, 7.8) rectangle (1.25, 8.4);

 \draw[thick] (-0.05, 8) -- (-0.25,8);
  		  \draw[thick] (-0.15, 7.95) -- (-0.15,0.05);
  		  \draw[thick] (-0.05, 0) -- (-0.25,0);
  		   \draw (-0.15, 4) node[anchor=east] {$C\cdot OPT$};

  		    \draw (0.625, 8.1) node {\Large $\mathbf{k}$};

		\end{tikzpicture}
		\hspace{20pt}
		\begin{tikzpicture}
		 
		 % \draw[thick] (6,0) rectangle (9,3);
  %\draw (6,0.4) rectangle (7.45,1);
           %\draw (6,1) rectangle (7.4,1.6);
           %\draw (6,1.6) rectangle (7.35,2.2);
    
          % \draw (7.8,1.8) rectangle (9,2.4);
  	      %  \draw[thick] (9.05, 1.8) -- (9.25,1.8);
  		  %\draw[thick] (9.15, 1.75) -- (9.15,0.05);
  		  %\draw[thick] (9.05, 0) -- (9.25,0);
  		   %\draw (9.4, 0.9) node {$h$};

  	      \draw[thick] (0,0) rectangle (3,3);
  	      
           \draw (0,0.4) rectangle (1.45,1);
           \draw (0,1) rectangle (1.4,1.6);
           \draw (0,1.6) rectangle (1.35,2.2);

           \draw (1.8,1.8) rectangle (3,2.4);
           \draw (1.85,1.2) rectangle (3,1.8);
           \draw (1.87,0.6) rectangle (3,1.2);

 \draw[thick] (-0.05, 0.4) -- (-0.25,0.4);
  		  \draw[thick] (-0.15, 0.35) -- (-0.15,0.05);
  		  \draw[thick] (-0.05, 0) -- (-0.25,0);
  		   \draw (-0.6, 0.2) node {\large $\le \frac{1}{9}$};

  		 \draw[thick] (-0.05, 3) -- (-0.25,3);
  		  \draw[thick] (-0.15, 2.95) -- (-0.15,2.25);
  		  \draw[thick] (-0.05, 2.2) -- (-0.25,2.2);
  		   \draw (-0.6, 2.6) node {\large $\le \frac{1}{9}$};
  		   
  		    \draw[thick] (3.05, 3) -- (3.25,3);
  		  \draw[thick] (3.15, 2.95) -- (3.15,2.45);
  		  \draw[thick] (3.05, 2.4) -- (3.25,2.4);
  		   \draw (3.6, 2.7) node {\large $\le \frac{1}{9}$};

    \draw[thick] (3.05, 0.6) -- (3.25,0.6);
  		  \draw[thick] (3.15, 0.55) -- (3.15,0.05);
  		  \draw[thick] (3.05, 0) -- (3.25,0);
  		   \draw (3.6, 0.3) node {\large $\le \frac{1}{9}$};
  	
		\end{tikzpicture}}\hspace{70pt}
        \scalebox{0.33}{
		\begin{tikzpicture}
		  \fill[color=lightgray!15] (0,6) rectangle (12,12);
          \fill[color=darkgray!45] (4,4) rectangle (6,6);
          \fill[color=gray!40] (6,0) rectangle (12,6);
          \fill[color=lightgray!40] (0,0) rectangle (4,6);
          \fill[color=lightgray!40] (4,0) rectangle (6,4);
          
            \draw[thick] (0,0) -- (12.5,0);
            \draw[thick] (0,0) -- (0,12.5);
            \draw[thick] (12,0) -- (12,-0.3);
            \draw[thick] (6,0) -- (6,-0.3);
            \draw[thick] (4,0) -- (4,-0.3);
            \draw[thick] (0,12) -- (-0.3,12);
            \draw[thick] (0,6) -- (-0.3,6);
            \draw[thick] (0,4) -- (-0.3,4);
            \draw[dashed] (0,12) -- (12,12);
            \draw[dashed] (12,0) -- (12,12);
            \draw[dashed] (0,6) -- (12,6);
            \draw[dashed] (6,0) -- (6,6);
            \draw[dashed] (4,4) -- (4,6);
            \draw[dashed] (4,4) -- (6,4);

            \draw (12,-0.3) node[anchor=north] {\LARGE $T$};
            \draw (6,-0.3) node[anchor=north] {\LARGE $\frac{T}{2}$};
            \draw (4,-0.3) node[anchor=north] {\LARGE $\frac{T}{3}$};
            \draw (-0.3,12) node[anchor=east] {\LARGE $C$};
            \draw (-0.3,6) node[anchor=east] {\LARGE $\frac{C}{2}$};
            \draw (-0.3,4) node[anchor=east] {\LARGE $\frac{C}{3}$};
            \draw (12.5,0) node[anchor=west] {\Large width};
            \draw (0,12.5) node[anchor=south] {\Large height};

            \draw (6,9) node {\huge \textbf{Tall}};
            \draw (3,3) node {\huge \textbf{Small}};
            \draw (5,5) node {\huge \textbf{Fat}};
            \draw (9,3) node {\huge \textbf{Wide}};
  		  
		\end{tikzpicture}}
		\caption{\textbf{(Left) } Construction of the $2$-structured solution from Lemma~\ref{lem:shortapx}. Dashed lines induce feasible allocations into bins, with an example of a bin next to it. \textbf{(Right) } Classification of the tasks in the $3$-approximation for the general case (Theorem~\ref{theo:3apx}).}
		\label{fig:short2apxf1}
			\end{figure} 

\subsection{2D Demand Bin Packing for square tasks}

In this section, we present a 2-approximation algorithm for square tasks and square bins (i.e. with $T=C$), and this late assumption will be removed later. 

\begin{theorem}\label{thm:2apx-square-square}
There is a $2$-approximation for 2D Demand Bin Packing for square tasks and square bins.
\end{theorem}
 
 \begin{proof} We will compute a $2$-structured solution for the tasks having height larger than $\frac{C}{4}$. Having that, we can conclude by applying Lemma~\ref{lem:smalltasks} to allocate the rest since $\left(1-\frac{1}{4}\right) \left(1-\frac{1}{4}\right) = \frac{9}{16} \ge \frac{1}{2}$. 

Consider first the set $I_1$ of tasks having height between $\frac{C}{4}$ and $\frac{C}{3}$. In any feasible solution, each bin contains at most $9$ such tasks, implying that $OPT \ge \lceil |I_1|/9 \rceil$. Since any set of $9$ tasks from $I_1$ fit into a bin, we can allocate these tasks into at most $\lceil |I_1|/9 \rceil$ bins via a simple iterative algorithm. For the set $I_2$ of tasks having height larger than $\frac{C}{3}$, we do the following: We sort the tasks non-increasingly by height, and place each task in any of the current bins as much to the left as possible; if the task does not fit in any bin, we open a new one and allocate the task there. Notice that in any feasible packing, each bin contains at most four of these tasks. If we denote by $k$ be the number of tasks in $I_2$ such that no other task fit together with them in a bin (i.e. tasks of height larger than $1-\min_{i\in I_2}{h_i}$), then $OPT \ge k + \lceil(|I_2|-k)/4\rceil$. 

One important property of the previous procedure is that if a bin has one task and the current task is allocated onto it because it fits, then any pair of subsequent tasks will also fit in the bin: it is possible to place them onto the corners of the bin due to the symmetric nature of their dimensions and then shift them to the left. This implies that whenever the procedure opens a new bin, the previous ones have either one or four tasks inside them. This property allows to prove its correctness as follows: Suppose by contradiction that the procedure opens more than $OPT$ bins, and let $i^*$ to be the first task allocated in the last open bin. As argued before every bin has either one or four tasks when the last bin was open by $i^*$, and the tasks assigned alone have height larger than $1-h_{i^*}$. This implies that if we consider the sequence $I'$ of tasks truncated on $i^*$, then the number of bins required to pack these tasks is at least the number of bins used by the procedure, which is a contradiction as the optimal solution for the general instance uses $OPT$ bins.

Putting the two solutions together, we obtain the desired $2$-structured solution. \end{proof}
 
It is possible to refine this algorithm to obtain a $2$-approximation for 2D Demand Bin Packing for square tasks with arbitrary (rectangular) bins (see Appendix~\ref{app:2apxsquare}).

\subsection{A simple $3$-approximation for 2D Demand Bin Packing}\label{sec:3-apx}

In this section, we present our $3$-approximation algorithm for 2D Demand Bin Packing. We will assume for simplicity that $T=C=1$ by scaling down the instance.

We will first describe a $4$-approximation for the problem, that will later be refined. Consider the optimal solution restricted to tasks with height larger than $\frac{1}{2}$. Since tall tasks cannot share slots, they induce a one dimensional Bin Packing instance defined by their widths that can be packed into at most $OPT$ bins. This means that we can obtain a solution for tall tasks that uses at most $\frac{3}{2}OPT$ bins by using the known $\frac{3}{2}$-approximation for Bin Packing~\cite{levi}. If we consider now the tasks having width larger than $\frac{1}{2}$ that are not large, we can proceed analogously as in every bin they must share the slot in position $\left\lceil\frac{T}{2}\right\rceil$, obtaining again a solution for them that uses at most $\frac{3}{2}OPT$ bins. Since the solution is $4$-structured, we can apply Lemma~\ref{lem:smalltasks} to get a solution that uses at most $4OPT$ bins as $\left( 1-\frac{1}{2}\right)\cdot\left( 1-\frac{1}{2}\right) \ge \frac{1}{4}$.

It can be observed that the solution for relatively large tasks is indeed $3$-structured, but Lemma~\ref{lem:smalltasks} might require extra bins to place relatively small tasks. In order to obtain an improved approximation, we will refine the classification of the tasks as follows (see Figure~\ref{fig:short2apxf1}): A task is \emph{tall} if $h_i > \frac{1}{2}$, \emph{wide} if $h_i \le \frac{1}{2}$ and $w_i > \frac{1}{2}$, \emph{fat} if $\frac{1}{3} < h_i \le \frac{1}{2}$ and $\frac{1}{3} < w_i \le \frac{1}{2}$, and \emph{small} otherwise, i.e. if $h_i \le \frac{1}{3}$ and $w_i\le \frac{1}{3}$. 

Notice that, if we manage to allocate the non-small tasks into at most $3OPT$ bins, there might still be tasks with height or width equal to $\frac{1}{2}$ preventing us from using directly Lemma~\ref{lem:smalltasks}. However, in this case, a task cannot have both height and width larger than $\frac{1}{3}$; the following slight adaptation of the aforementioned lemma will be useful to attain the result.

\begin{lemma}\label{lem:smalltasksm} 
  If it is possible to compute a $3$-structured solution for wide, tall, and fat tasks, then it is possible to compute a feasible solution for the whole instance that uses at most $3OPT$ bins.
 \end{lemma}

 \begin{proof} Let $\ell \le 3OPT$ denote the number of $\frac{1}{3}$-full bins from the initial packing. We will use again the First-Fit inspired algorithm to pack the remaining tasks into the bins with non-increasing load profiles (possibly opening new bins in the process). Suppose by contradiction that we used more than $3OPT$ bins in total. This implies that the tasks assigned to the last bin are all small. If one of the tasks has height and width at most $\frac{1}{3}$ then as in Lemma~\ref{lem:smalltasks}, each one of these bins (except for the new opened one) must have total load at least $\left(1-\frac{1}{3}\right)\left(1-\frac{1}{3}\right) \ge \frac{1}{3}$ reaching a contradiction analogously to the proof of Lemma~\ref{lem:smalltasks}. If on the other hand the tasks in the last bin have height (resp. width) larger than $\frac{1}{3}$, then their width (resp. height) must be at most $\frac{1}{3}$, which then again as in Lemma~\ref{lem:smalltasks} that each of the bins that were initially not $\frac{1}{3}$-full have now total load at least $\left(1-\frac{1}{2}\right)\left(1-\frac{1}{3}\right) \ge \frac{1}{3}$, reaching again a contradiction. \end{proof}
 
From now on, we will focus on constructing $3$-structured solutions for non-small tasks, considering two cases depending on the value $OPT$.
 
\subsubsection{The case $OPT \le 70$}

Let us start with the case $OPT\le 70$. As described in the initial $4$-approximation, tall and wide tasks induce Bin Packing instances. This time, instead of directly applying the $\frac{3}{2}$-approximation to solve them, we will first study the interplay between fat, tall and wide tasks in any feasible allocation, as intuitively a bin should not have too many fat tasks while having large total area of tall and wide tasks. 

\begin{lemma}\label{lem:3fat} Consider a feasible allocation of tasks into a bin.

\begin{enumerate}
    \item If it contains three fat tasks, and tall tasks of total width larger than $\frac{8}{9}$, then the total height of wide tasks is smaller than $\frac{4}{9}$. A symmetric statement holds interchanging the roles of tall and wide tasks.
    \item If it contains four fat tasks, then the total width of tall tasks and the total height of wide tasks is at most $\frac{2}{3}$.
\end{enumerate} \end{lemma}

\begin{proof} \begin{enumerate}
\item We will prove the first statement of the lemma as the second one is completely symmetric. The total area of fat tasks in the bin is at least $3\cdot\frac{1}{9} = \frac{1}{3}$ and the total area of tall tasks in the bin is larger than $\frac{1}{2}-\frac{1}{18}$, implying that the total area of wide tasks in the bin is smaller than $\frac{2}{9}$. Hence, their total height must be smaller than $\frac{4}{9}$.

\item Consider first the tall tasks in the bin. Since the bin contains four fat tasks, they must overlap and share slots of total width at least $\frac{1}{3}$. Since these slots would have total load of fat tasks at least $\frac{2}{3}$, no tall task can use that slot and then their total width is bounded by $\frac{2}{3}$ as tall tasks cannot share slots.
	
  Regarding wide tasks, we can assume that no fat task intersects the vertical line in the middle of the bin because every wide task intersects this line and then the claim would be already fulfilled. This implies that two of the fat tasks must be in the left half of the bin and the other two in the right half. Because of these, fat tasks on the left side must occupy the slot in position $\frac{1}{3}$ and fat tasks on the right side must occupy the slot in position $\frac{2}{3}$. Notice also that each wide task must occupy the slot in position $\frac{1}{3}$ or the one in position $\frac{2}{3}$ as the space between these two is only $\frac{1}{3}$. This already implies that the total height of wide tasks is at most the free load left by fat tasks in the slot in position $\frac{1}{3}$ plus the free load left by fat tasks in the slot in position $\frac{2}{3}$, which is at most $\frac{2}{3}$. \end{enumerate} \end{proof}
  
  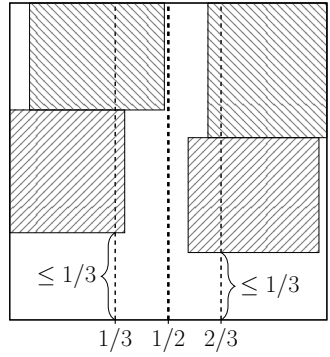
\begin{figure}
	    \centering
    \resizebox{0.26\textwidth}{!}{
		\begin{tikzpicture}
		  \draw[thick] (0,0) rectangle (8,8);
  		  \draw[dashed] (2.66, 0) -- (2.66, 8);
  		   \draw[dashed] (5.33, 0) -- (5.33, 8);
  		  \draw[dashed, ultra thick] (4, 0) -- (4, 8);
  		  
  		  \draw[fill=gray!40, pattern=north east lines, pattern color=gray] (0,2.2) rectangle (2.9,5.3);
  		  \draw[fill=gray!40, pattern=north west lines, pattern color=gray] (0.5,5.3) rectangle (3.9,8);
  		  
  		  \draw[fill=gray!40, pattern=north east lines, pattern color=gray] (4.5,1.7) rectangle (7.8,4.6);
  		  \draw[fill=gray!40, pattern=north west lines, pattern color=gray] (5,4.6) rectangle (8,8);
  		  
  		  \draw [decorate,decoration={brace,amplitude=10pt},xshift=0pt,yshift=0pt]
(2.66,0) -- (2.66,2.2) node [black,midway,xshift=-35pt] 
{\LARGE $\le 1/3$};

\draw [decorate,decoration={brace,amplitude=10pt},xshift=0pt,yshift=0pt]
(5.33,1.7) -- (5.33,0) node [black,midway,xshift=35pt] 
{\LARGE $\le 1/3$};

           \draw[thick] (2.66, 0) -- (2.66,-0.1);
			\draw[thick] (5.33, 0) -- (5.33,-0.1);
           \draw[thick] (4, 0) -- (4,-0.1);
            \draw (2.66, -0.5) node {\LARGE $1/3$};
            \draw (5.33, -0.5) node {\LARGE $2/3$};
		    \draw (4, -0.5) node {\LARGE $1/2$};
		
		\draw[thick] (0,0) rectangle (8,8);
  		  \draw[dashed] (2.66, 0) -- (2.66, 8);
  		   \draw[dashed] (5.33, 0) -- (5.33, 8);
  		  \draw[dashed, ultra thick] (4, 0) -- (4, 8);   
  		  
		\end{tikzpicture}}
		\caption{A feasible allocation with four fat tasks that do not occupy the middle slot as in the proof of Lemma~\ref{lem:3fat}. Their placement does not allow to have wide rectangles of arbitrarily large total height in the bin.}
			\label{fig:3apx4fatwide}
	\end{figure}

  Having such stronger area guarantees for each type of task is helpful, as then we can use the PTAS for \emph{Multiple Knapsack} due to Chekuri and Khanna~\cite{multnap} to obtain efficient solutions in terms of bins. In the latter problem, we are given a set of one-dimensional bins, each one having a given capacity, and a set of items having given sizes and profits, and the goal is to assign a subset of the items of maximum profit to the bins so that the total size of the items in a bin is at most its capacity.
  
\begin{lemma}\label{lem:MKPpacking}
Consider a set of items, where each item $i$ is characterized by its size $s_i$, and a set of $K$ bins, $K\in \mathbb{N}$ constant, where each bin $j$ is characterized by its capacity $c_j\le 1$. If there exists a feasible packing of all the items into the bins and $\sum_i{s_i} \le \left(\sum_j{c_j}\right)-B$ for some constant $B>0$, then it is possible to pack all the items into the $K$ bins in polynomial time. \end{lemma}

\begin{proof} Since $\sum_i{s_i} \le \left(\sum_j{c_j}\right)-B$, there must be a bin $j'$ such that its free capacity is at least $\frac{B}{K}$ (we can assume $j'$ is known by guessing). Consider the Multiple Knapsack instance induced by the remaining $K-1$ bins and the items in the instance, having size and profit equal to $s_i$. Since there exists a feasible solution for this instance of total profit at least $\sum_i{s_i} - \left(c_{j'} - \frac{B}{K}\right)$, we can apply the $(1+\varepsilon)$-approximation for Multiple Knapsack with parameter $\varepsilon=\frac{B}{K^2}$ to obtain a packing of a subset of the items into these $K-1$ bins of total size at least $(1-\varepsilon)\left(\sum_i{s_i}-c_{j'}-\frac{B}{K}\right)$. We conclude by noticing that the total area of the items that were not packed is at most \begin{eqnarray*} \sum_i{s_i} - (1-\varepsilon)\left(\sum_i{s_i}-c_{j'}-\frac{B}{K}\right) & = & c_{j'} - \frac{B}{K} + \varepsilon \sum_i{s_i} - \varepsilon c_{j'} + \varepsilon \frac{B}{K} \\ & \le & c_{j'} - \frac{B}{K} + \varepsilon K \le c_{j'}, \end{eqnarray*} meaning that we can pack them into bin $j'$. \end{proof}

Lemma~\ref{lem:MKPpacking} implies, in particular, that (one-dimensional) Bin Packing instances induced by tall (resp. wide) tasks that fit into $K$ bins, can indeed be packed into $K$ bins in polynomial time if their total width (resp. height) is smaller than $K$ by a constant amount. We will use this result to prove that, when $OPT\le 70$, it is possible to compute a feasible well-structured solution of the tall, wide and fat tasks.

\begin{lemma}\label{lem:opt<70} For instances of 2D Demand Bin Packing satisfying that $OPT\le 70$, there exists a polynomial-time algorithm that computes a $3$-structured solution for tall, wide, and fat tasks that uses at most $3OPT$ bins. \end{lemma}

\begin{proof} Consider an instance of 2D Demand Bin Packing with optimal number of bins $OPT\le 70$. Let $\mathcal{T}$ be the set of tall tasks and $\mathcal{W}$ the set of wide tasks in the instance. We will consider four cases depending on the total width of tall tasks $w(\mathcal{T})$ and total height of wide tasks $h(\mathcal{W})$ in the instance.

\begin{itemize}
    \item If $h(\mathcal{W})>OPT-\frac{1}{9}$ and $w(\mathcal{T})>OPT-\frac{1}{9}$, then the total area of these tasks is larger than $OPT-\frac{1}{9}$ and in particular there are no fat tasks in the instance. Hence we can simply use the $\frac{3}{2}$-approximation for Bin Packing in the instances these tasks induce separately and obtain a solution for them using at most $3OPT$ bins satisfying the required properties (since we can sort the tasks in each bin according to their width or height, a sorted load profile can be easily attained).
    \item If $h(\mathcal{W})\le OPT-\frac{1}{9}$ and $w(\mathcal{T})\le OPT-\frac{1}{9}$ then we can use Lemma~\ref{lem:MKPpacking} to obtain a solution for tall tasks using at most $OPT$ bins, and a solution for wide tasks using at most $OPT$ bins. Notice that fat tasks can be allocated using NFDH into at most $OPT$ bins as any set of four fat tasks fit into a bin. The obtained solution satisfies the required properties.
    \item If $h(\mathcal{W})>OPT-\frac{1}{9}$ and $w(\mathcal{T})\le OPT-\frac{1}{9}$, then in each bin the total height of wide tasks is larger than $1-\frac{1}{9}$. Thanks to Lemma~\ref{lem:3fat}, we know then that each bin has at most two fat tasks or the total width of tall tasks in the bin is at most $\frac{1}{2}-\frac{1}{18}$. Let $OPT_1$ be the number of bins with at most two fat tasks and $OPT_2$ the number of remaining bins, and notice that each bin in the second group contains at most three fat tasks thanks to Lemma~\ref{lem:3fat}. We will use the $\frac{3}{2}$-approximation for Bin Packing to allocate tall tasks, but for wide tasks we will use instead Lemma~\ref{lem:MKPpacking} with $OPT_1$ bins of capacity $1$ and $OPT_2$ bins of capacity $\frac{1}{2}$ (we guess these two values). Notice that the hypotheses of the lemma are fulfilled as if $OPT_2>1$ then the total height of the wide tasks differs from the total capacity of the bins by at least $\frac{1}{18}$, and if $OPT_2=0$ this difference is at least $\frac{1}{9}$. We will allocate the fat tasks as follows: In each of the $OPT_2$ bins that contain wide tasks of total height at most $\frac{1}{2}$ we will allocate any two fat tasks, and the remaining ones (if any) will be placed into new bins in groups of four using NFDH. This way we obtain a feasible solution with only non-increasing profiles (bins containing wide and fat tasks can be rearranged in that way as tasks are placed in shelves) whose number of bins is at most \[\left\lfloor \frac{3}{2}OPT\right\rfloor + OPT + \left\lceil \frac{2OPT_1 + OPT_2}{4} \right\rceil \le \left\lfloor \frac{3}{2}OPT\right\rfloor + OPT + \left\lceil \frac{OPT}{2} \right\rceil = 3OPT.\]
    \item If $h(\mathcal{W})\le OPT-\frac{1}{9}$ and $w(\mathcal{T})\le OPT-\frac{1}{9}$ we can do a symmetric procedure as in the previous case to obtain $OPT_1$ bins of width $1$ and $OPT_2$ bins of width $\frac{1}{2}$ for tall tasks, but we need a slight modification in the assignment of fat tasks to ensure that the obtained load profiles are non-increasing. We will place the $OPT_2$ bins of width $\frac{1}{2}$ in groups of $2$ into $\lceil OPT_2/2\rceil$ bins, and if $OPT_2$ is odd we will place one fat task in the bin that tall tasks of total width at most $\frac{1}{2}$ in it. The remaining fat tasks are allocated using NFDH into new bins. This way we obtain a solution with non-increasing load profiles simply by sorting the tasks non-increasingly by height in bins containing tall tasks. The number of bins in this solution is at most \[\left\lfloor \frac{3}{2}OPT\right\rfloor + OPT + \left\lceil \frac{2OPT_1 + OPT_2}{4} \right\rceil \le \left\lfloor \frac{3}{2}OPT\right\rfloor + OPT + \left\lceil \frac{OPT}{2} \right\rceil = 3OPT\] if $OPT_2$ is even, and at most \[\left\lfloor \frac{3}{2}OPT\right\rfloor + OPT + \left\lceil \frac{2OPT_1 + OPT_2+1}{4} \right\rceil \le \left\lfloor \frac{3}{2}OPT\right\rfloor + OPT + \left\lceil \frac{OPT}{2} \right\rceil = 3OPT\] if $OPT_2$ is odd.
\end{itemize}\end{proof}

\subsubsection{The case $OPT > 70$}
   For large values of $OPT$, we can instead use good asymptotic approximation algorithms for (one-dimensional) Bin Packing for the induced subroutines, such as the well known APTAS that uses at most $(1 + \varepsilon)OPT + 1$ bins~\cite{VL81}. If $\varepsilon = \frac{1}{70}$, then in this case $(1 + \varepsilon)OPT + 1 \le (1+2\varepsilon)OPT$.
   
\begin{lemma}\label{lem:opt>70} If $OPT> 70$, there is a polynomial-time algorithm that computes a $3$-structured solution for tall, wide, and fat tasks that uses at most $3OPT$ bins. \end{lemma}
   
\begin{proof} We will divide the proof into two cases depending on the number of fat tasks in the instance. If is is at most $4(1-4\varepsilon)OPT$, then we can place them into at most $(1-4\varepsilon)OPT$ bins as any four such tasks fit together (and we can use NFDH for each such group to attain sorted load profiles). Using the APTAS for Bin Packing~\cite{VL81}, we can place the tall tasks and the wide tasks into at most $(1 + 2\varepsilon)OPT$ bins each respectively, obtaining at most $3OPT$ bins as required.
   
If the number of fat tasks is larger than $4(1-4\varepsilon)OPT$, which we can place into $OPT$ bins as before, the total area of wide and tall tasks is at most $\left(\frac{5}{9} + \frac{16}{9} \varepsilon\right) OPT$. We can place the wide tasks having height larger than $\frac{1}{3}$ and the tall tasks having width larger than $\frac{1}{3}$ into $16\varepsilon OPT$ separate empty bins as there cannot be more than $4OPT$ tasks having height and width larger than $\frac{1}{3}$. The remaining tall tasks will be allocated greedily one next to the other, and wide tasks as well placing them one on top of the other. Since the width (resp. height) of these tall (resp. wide) tasks is at most $\frac{1}{3}$ and their height (resp. width) is at least $\frac{1}{2}$, each bin except possibly for the last one will have total area of tasks at least $\frac{1}{3}$. Hence, we would have used at most $3\left(\frac{5}{9}OPT + \frac{16}{9} \varepsilon OPT\right) + 2$ bins for them. The total number of bins is at most $OPT + 16\varepsilon OPT + 3\left(\frac{5}{9}OPT + \frac{16}{9} \varepsilon OPT\right) + 2 \le 3OPT$, and the structure of the load profiles can be ensured simply by sorting the tasks in each bin. \end{proof}

Putting Lemmas~\ref{lem:smalltasksm}, \ref{lem:opt<70}, and \ref{lem:opt>70} together, we conclude the claimed result.
\begin{theorem}\label{theo:3apx} There is a 3-approximation algorithm for 2D Demand Bin Packing problem. \end{theorem}
\vspace{-5pt}
\section{Concluding remarks}

We presented hardness results and approximation algorithms for a two-dimensional version of Bin Packing denoted as 2D Demand Bin Packing. The main open question is to devise a tight $2$-approximation algorithm for the problem, which is known to exist for 2D Geometric Bin Packing~\cite{2d}. As it was the case for the latter work, the core of the problem most likely lies in the case when $OPT=1$; our current approach considers more or less independently the three sets of tall, wide and fat tasks, which would not be enough to achieve such a result. Further understanding that special case would be a big step towards a tight approximation algorithm for our problem. Another interesting open question is the case of Asymptotic approximation algorithms, where hopefully recent results for the geometric case can be simplified or even improved.

\bibliographystyle{plainurl}
\bibliography{bibliography}

\begin{thebibliography}{10}

\bibitem{ABCGJKLP13}
Soroush Alamdari, Therese~C. Biedl, Timothy~M. Chan, Elyot Grant, Krishnam~Raju
  Jampani, Srinivasan Keshav, Anna Lubiw, and Vinayak Pathak.
\newblock Smart-grid electricity allocation via strip packing with slicing.
\newblock In {\em {WADS} 2013}, volume 8037, pages 25--36, 2013.
\newblock \href {https://doi.org/10.1007/978-3-642-40104-6\_3}
  {\path{doi:10.1007/978-3-642-40104-6\_3}}.

\bibitem{BEK16}
Nikhil Bansal, Marek Eli{\'{a}}s, and Arindam Khan.
\newblock Improved approximation for vector bin packing.
\newblock In {\em {SODA} 2016}, pages 1561--1579, 2016.
\newblock \href {https://doi.org/10.1137/1.9781611974331.CH106}
  {\path{doi:10.1137/1.9781611974331.CH106}}.

\bibitem{BK14}
Nikhil Bansal and Arindam Khan.
\newblock Improved approximation algorithm for two-dimensional bin packing.
\newblock In {\em {SODA} 2014}, pages 13--25, 2014.
\newblock \href {https://doi.org/10.1137/1.9781611973402.2}
  {\path{doi:10.1137/1.9781611973402.2}}.

\bibitem{multnap}
Chandra Chekuri and Sanjeev Khanna.
\newblock A polynomial time approximation scheme for the multiple knapsack
  problem.
\newblock {\em {SIAM} J. Comput.}, 35(3):713--728, 2005.
\newblock \href {https://doi.org/10.1137/S0097539700382820}
  {\path{doi:10.1137/S0097539700382820}}.

\bibitem{CC09}
Miroslav Chleb{\'{\i}}k and Janka Chleb{\'{\i}}kov{\'{a}}.
\newblock Hardness of approximation for orthogonal rectangle packing and
  covering problems.
\newblock {\em J. Discrete Algorithms}, 7(3):291--305, 2009.
\newblock \href {https://doi.org/10.1016/J.JDA.2009.02.002}
  {\path{doi:10.1016/J.JDA.2009.02.002}}.

\bibitem{CKPT17}
Henrik~I. Christensen, Arindam Khan, Sebastian Pokutta, and Prasad Tetali.
\newblock Approximation and online algorithms for multidimensional bin packing:
  {A} survey.
\newblock {\em Comput. Sci. Rev.}, 24:63--79, 2017.
\newblock \href {https://doi.org/10.1016/J.COSREV.2016.12.001}
  {\path{doi:10.1016/J.COSREV.2016.12.001}}.

\bibitem{survey}
Edward~G. Coffman~Jr., J{\'a}nos Csirik, G{\'a}bor Galambos, Silvano Martello,
  and Daniele Vigo.
\newblock {\em Bin Packing Approximation Algorithms: Survey and
  Classification}.
\newblock 2013.
\newblock \href {https://doi.org/10.1007/978-1-4419-7997-1\_35}
  {\path{doi:10.1007/978-1-4419-7997-1\_35}}.

\bibitem{VL81}
Wenceslas~Fernandez de~la Vega and George~S. Lueker.
\newblock Bin packing can be solved within 1+epsilon in linear time.
\newblock {\em Comb.}, 1(4):349--355, 1981.
\newblock \href {https://doi.org/10.1007/BF02579456}
  {\path{doi:10.1007/BF02579456}}.

\bibitem{DIM16}
Maxence Delorme, Manuel Iori, and Silvano Martello.
\newblock Bin packing and cutting stock problems: Mathematical models and exact
  algorithms.
\newblock {\em Eur. J. Oper. Res.}, 255(1):1--20, 2016.
\newblock \href {https://doi.org/10.1016/J.EJOR.2016.04.030}
  {\path{doi:10.1016/J.EJOR.2016.04.030}}.

\bibitem{EHRW25}
Franziska Eberle, Felix Hommelsheim, Malin Rau, and Stefan Walzer.
\newblock A tight $(3/2 + \varepsilon)$-approximation algorithm for demand
  strip packing.
\newblock In {\em {SODA} 2025}, pages 641--699, 2025.
\newblock \href {https://doi.org/10.1137/1.9781611978322.20}
  {\path{doi:10.1137/1.9781611978322.20}}.

\bibitem{GJ79}
M.~R. Garey and David~S. Johnson.
\newblock {\em Computers and Intractability: {A} Guide to the Theory of
  NP-Completeness}.
\newblock 1979.

\bibitem{GGJK21}
Waldo Gálvez, Fabrizio Grandoni, Afrouz~Jabal Ameli, and Kamyar Khodamoradi.
\newblock Approximation algorithms for demand strip packing.
\newblock In {\em {APPROX} 2021}, volume 207, pages 20:1--20:24, 2021.
\newblock \href {https://doi.org/10.4230/LIPICS.APPROX/RANDOM.2021.20}
  {\path{doi:10.4230/LIPICS.APPROX/RANDOM.2021.20}}.

\bibitem{2d}
Rolf Harren, Klaus Jansen, Lars Pr{\"{a}}del, Ulrich~M. Schwarz, and Rob van
  Stee.
\newblock Two for one: Tight approximation of 2d bin packing.
\newblock {\em Int. J. Found. Comput. Sci.}, 24(8):1299--1328, 2013.
\newblock \href {https://doi.org/10.1142/S0129054113500354}
  {\path{doi:10.1142/S0129054113500354}}.

\bibitem{HJPS14}
Rolf Harren, Klaus Jansen, Lars Pr{\"{a}}del, and Rob van Stee.
\newblock A {(5/3} + {\(\epsilon\)})-approximation for strip packing.
\newblock {\em Comput. Geom.}, 47(2):248--267, 2014.
\newblock \href {https://doi.org/10.1016/J.COMGEO.2013.08.008}
  {\path{doi:10.1016/J.COMGEO.2013.08.008}}.

\bibitem{HR17}
Rebecca Hoberg and Thomas Rothvoss.
\newblock A logarithmic additive integrality gap for bin packing.
\newblock In {\em {SODA} 2017}, pages 2616--2625, 2017.
\newblock \href {https://doi.org/10.1137/1.9781611974782.172}
  {\path{doi:10.1137/1.9781611974782.172}}.

\bibitem{KK82}
Narendra Karmarkar and Richard~M. Karp.
\newblock An efficient approximation scheme for the one-dimensional bin-packing
  problem.
\newblock In {\em {FOCS} 1982}, pages 312--320. {IEEE} Computer Society, 1982.
\newblock \href {https://doi.org/10.1109/SFCS.1982.61}
  {\path{doi:10.1109/SFCS.1982.61}}.

\bibitem{KMS23}
Ariel Kulik, Matthias Mnich, and Hadas Shachnai.
\newblock Improved approximations for vector bin packing via iterative
  randomized rounding.
\newblock In {\em {FOCS} 2023}, pages 1366--1376, 2023.
\newblock \href {https://doi.org/10.1109/FOCS57990.2023.00084}
  {\path{doi:10.1109/FOCS57990.2023.00084}}.

\bibitem{leung}
Joseph~Y.{-}T. Leung, Tommy~W. Tam, C.~S. Wong, Gilbert~H. Young, and Francis
  Y.~L. Chin.
\newblock Packing squares into a square.
\newblock {\em J. Parallel Distributed Comput.}, 10(3):271--275, 1990.
\newblock \href {https://doi.org/10.1016/0743-7315(90)90019-L}
  {\path{doi:10.1016/0743-7315(90)90019-L}}.

\bibitem{levi}
David Simchi-Levi.
\newblock New worst-case results for the bin-packing problem.
\newblock {\em Naval Research Logistics (NRL)}, 41(4):579--585, 1994.
\newblock \href
  {https://doi.org/10.1002/1520-6750(199406)41:4<579::AID-NAV3220410409>3.0.CO;2-G}
  {\path{doi:10.1002/1520-6750(199406)41:4<579::AID-NAV3220410409>3.0.CO;2-G}}.

\bibitem{THLW13}
Shaojie Tang, Qiuyuan Huang, Xiang{-}Yang Li, and Dapeng Wu.
\newblock Smoothing the energy consumption: Peak demand reduction in smart
  grid.
\newblock In {\em {INFOCOM} 2013}, pages 1133--1141, 2013.
\newblock \href {https://doi.org/10.1109/INFCOM.2013.6566904}
  {\path{doi:10.1109/INFCOM.2013.6566904}}.

\bibitem{Z05}
Guochuan Zhang.
\newblock A 3-approximation algorithm for two-dimensional bin packing.
\newblock {\em Oper. Res. Lett.}, 33(2):121--126, 2005.
\newblock \href {https://doi.org/10.1016/J.ORL.2004.04.004}
  {\path{doi:10.1016/J.ORL.2004.04.004}}.

\end{thebibliography}

\clearpage
\appendix

\section{Hardness of Approximation for variants of 2D Demand Bin Packing}\label{app:hardness}

In this section, we will prove that deciding whether an instance of 2D Demand Bin Packing can be allocated into one bin is strongly NP-hard in two restricted cases: when the height of the tasks is small compared to the capacity $C$, and when the tasks and the bin are squares (i.e., $T=C$ and, for each task $i$, $h_i=w_i$). This implies a hardness of approximation of $(2-\varepsilon)$ for the respective variants of 2D Demand Bin Packing, implying as well that the presented results are best possible (see Theorems~\ref{lem:shortapx}, \ref{thm:2apx-square-square} and \ref{thm:2apx-square-gen}). 

The decision problem, denoted as \emph{Demand Allocation problem}, is defined as follows: Given a timeline of length $T$ with capacity $C$ on each time slot, and $n$ tasks, where each task $i$ has height $h_i$, is there a feasible allocation of the tasks into the timeline?.

In both cases, we will provide a reduction from the strongly NP-hard $3$-Partition problem. In this problem, we receive as instance a set of $3n$ numbers $\{a_1,\dots,a_{3n}\}$, satisfying that $\frac{\sum_j{a_j}}{4} < a_i < \frac{\sum_j{a_j}}{2}$ for every $j=1,\dots,3n$, and the goal is to decide whether there exists a partition of the set into $n$ sets of three numbers each, such that the total sum of the numbers in each set is the same. 

\subsection{Hardness of Approximation for 2D Demand Bin Packing of Short Tasks}

We start by considering the case of short tasks.

\begin{lemma}\label{lem:hardness_short}
 The Demand Allocation problem is strongly NP-hard even when restricted to tasks of height at most $3/n$ unless P$=$NP, where $n$ is the number of tasks in the instance.
\end{lemma}

\begin{proof}
Let $\{a_1,...,a_{3z}\}$ be an arbitrary instance of 3-Partition. We construct a Demand Allocation instance with $3z$ tasks, where each task $i$ will have height $1$ and width $a_i$, and the bins will have capacity $C=z$ and width $T = \frac{1}{z}\sum_{i=1}^{3z}{a_i}$. If there exists a $3$-partition, then there exists an allocation of the tasks into the bin as we can pack the sets of the partition into shelves of height $1$. On the other hand, if there exists an allocation of the tasks into the bin, then the bin is full as the total area of the tasks is $z\cdot T$. We can derive a $3$-Partition from the optimal allocation as follows:

Since the bin is full, there must be $z$ tasks starting at the leftmost slot of the bin. Consider one such task $i_1$. Since the bin is full, some task must start where $i_1$ ends. Let us call $i_2$ the task starting there. Notice that $w_{i_1} + w_{i_2} < W$, so analogously some task must start where $i_2$ ends. Let us call it $i_3$. Notice that no task can start after $i_3$ ends as then the total width of the four tasks would be larger than $T$. Hence the total width of $i_1, i_2$ and $i_3$ is exactly $T$ since the bin is full. If we remove these tasks we are left with a full bin of height $z-1$ and width $T$. So we can iterate the argument and obtain a $3$-partition. \end{proof}

%Another interesting special case is the one of \emph{square tasks}, meaning that each task $i$ satisfies $h(i)=w(i)$. By slightly adapting a result from Leung et al.~\cite{leung} for a related geometric version of the problem, it is possible to prove that the same hardness of approximation holds even for this case.

\begin{corollary}
 For any $\varepsilon > 0$, there is no $(2-\varepsilon)$-approximation algorithm for 2D Demand Bin Packing problem for tasks of height $3/n$ unless P$=$NP, where $n$ is the number of tasks in the instance.
\end{corollary}

\begin{proof}
Let us assume that we have a $(2-\varepsilon)$-approximation algorithm for 2D Demand Bin Packing for tasks of height at most $3/n$. Given an instance of Demand Allocation with such short tasks, we can decide whether $OPT=1$ by applying this algorithm since a $(2-\varepsilon)$-approximation returns the optimal solution when $OPT=1$. This means that we can solve Demand Allocation restricted to tasks of height at most $3/n$ in polynomial time, a contradiction to Lemma~\ref{lem:hardness_short} unless P=NP. \end{proof}

\subsection{Hardness of Approximation for 2D Demand Bin Packing of Square Tasks into Square Bins}

We continue with the case of square tasks. The decision problem is defined as follows: Given a timeline of length $N$ with capacity $N$ on each time slot, and $n$ square tasks, where each task $i$ has length side $\ell_i$, is there a feasible allocation of the tasks into the timeline?.

\begin{lemma}\label{lem:hardness-leung} The Demand Allocation problem is strongly NP-hard even when restricted to square tasks and square bins, unless P$=$NP. \end{lemma}

Our proof is an adaptation of the construction due to Leung et al.~\cite{leung} for the case of geometric square packing; as mentioned before, the hardness does not extend directly as there exist instances that cannot be packed geometrically but that admit a feasible allocation (see Appendix~\ref{app:gap_instance}).

%We will prove this result via a reduction from $3$-Partition, which is known to be strongly NP-hard. In the 3-Partition problem, we receive as instance a set of $3n$ numbers $\{a_1,\dots,a_{3n}\}$, satisfying that $\frac{\sum_j{a_j}}{4} < a_i < \frac{\sum_j{a_j}}{2}$ for every $j=1,\dots,3n$, and the goal is to decide whether there exists a partition of the set into $n$ sets of three numbers each, such that the total sum of the numbers in each set is the same. 

Let $I = \{a_1,\dots,a_{3n}\}$ be an instance of the $3$-Partition problem, and let $B=\frac{1}{n}\sum_i{a_i}$. We will assume that $n$ is a multiple of $3$ and that the equation $x(x+1)-1=n$ admits an integer solution $x^*$; if it is not the case, we can add triplets of dummy numbers to the 3-Partition instance of the form $\frac{B}{2}-2\varepsilon$, $\frac{B}{4}+\varepsilon$, and $\frac{B}{4}+\varepsilon$, where $\varepsilon>0$ is a small enough constant so that $\frac{B}{4}+\varepsilon < a_i$ for any $i=1,\dots,3n$. It is not difficult to see that the original instance admits a $3$-Partition if and only if the new instance admits it, since the dummy numbers of the form $\frac{B}{2}-2\varepsilon$ can only be paired with the dummy numbers of the form $\frac{B}{4}+\varepsilon$ so as to sum up to $B$ by construction. We will construct an instance of the Demand Square Allocation problem as follows, where $D=3B^3+B$:

\begin{itemize}
    \item We define a timeline of length $N = x^*(x^*+2)D = (n+x^*+1)D$, and the capacity of each time slot consequently is also $N$.
    \item We define $2x^*+3$ square tasks, denoted as \emph{Structure tasks}, such that one of them has side length $nD$, $x^*$ of them have side length $(x^*+1)D$, and $x^*+2$ of them have side length $x^*D$.
    \item We define $3n$ square tasks, denoted as \emph{Partition tasks}, such that for each $i=1\dots,3n$, there is a square task of side length $B^3 + a_i$.
    \item We define $\frac{2n}{3}$ square tasks, denoted as \emph{Enforcer tasks}, each one of side length $3B^3$.
\end{itemize}

As proved by Leung et al.~\cite{leung}, if the $3$-Partition instance admits a partition, then there exists a geometric packing (and hence a feasible allocation) for all the tasks into a bin. See Figure~\ref{fig:leung-packing} for a depiction of the solution.

\begin{figure}
	    \centering
    
		\resizebox{0.47\textwidth}{!}{\begin{tikzpicture}
            \draw[dashed] (1.5,7) rectangle (5.5,8.5);
		  \draw[thick] (0,0) rectangle (8.5,8.5);
		  \draw[thick] (0,0) rectangle (6.5,6.5);
		  \draw[thick] (6.5,0) rectangle (8.5,2);
		   \draw[thick] (0,8.5) rectangle (1.5,7);
		    \draw[thick] (5.5,8.5) rectangle (7,7);
		   \draw[thick] (8.5,8.5) rectangle (7,7);
            \draw[thick] (6.5,7) rectangle (8.5,5);
            
              \draw (7.5, 1) node {Structure};
              \draw (7.5, 6) node {Structure};
              \draw (7.5, 3.75) node {$.$};
              \draw (7.5, 3.25) node {$.$};
              \draw (7.5, 3.5) node {$.$};
              \draw (0.75, 7.75) node {Structure};
              \draw (6.25, 7.75) node {Structure};
              \draw (7.75, 7.75) node {Structure};
              \draw (3.25, 7.75) node {$.$};
              \draw (3.5, 7.75) node {$.$};
              \draw (3.75, 7.75) node {$.$};
  			 \draw (3.25, 3.25) node {Structure};
  			 \draw (3.25, 6.75) node {Enforcer and Partition tasks};

            \draw[thick] (-0.15, 6.55) -- (-0.15, 6.95);
           \draw[thick] (-0.05, 6.5) -- (-0.25,6.5);
           \draw[thick] (-0.05,7) -- (-0.25,7);
            \draw (-0.15,6.75) node[anchor=east] {$D$} ;
            
 \draw[thick] (0.05,-0.15) -- (6.45,-0.15);
           \draw[thick] (0,-0.05) -- (0,-0.25);
           \draw[thick] (6.5,-0.05) -- (6.5,-0.25);
            \draw (3.25, -0.4) node {$nD$};

        \draw[thick] (6.55,-0.15) -- (8.45,-0.15);
           \draw[thick] (6.5,-0.05) -- (6.5,-0.25);
           \draw[thick] (8.5,-0.05) -- (8.5,-0.25);
            \draw (7.5, -0.4) node {$(x^*+1)D$};

		\end{tikzpicture}}\hspace{30pt}
        \resizebox{0.45\textwidth}{!}{\begin{tikzpicture}
            \draw[dashed] (0,0) rectangle (8,1.5);
		  \draw[thick] (0,0) rectangle (1.25,1.25);
            \draw (0.625,0.625) node {\small Enforcer} ;
            \draw (2.625,0.625) node {$\dots$};
		  \draw[thick] (4,0) rectangle (5.25,1.25);
            \draw (4.625,0.625) node {\small Enforcer} ;
            
            \draw[thick] (5.25,0) rectangle (5.75,1.5);
            \draw (5.5,0.75) node[rotate=90] {\small Partition 1} ;
            \draw[thick] (5.75,0) rectangle (6.25,1.5);
            \draw (6,0.75) node[rotate=90] {\small Partition 2} ;
            \draw[thick] (7.4,0) rectangle (7.9,1.5);
            \draw (7.65,0.75) node[rotate=90] {\small Partition n} ;
            \draw (6.825,0.75) node {$\dots$};

            \draw[thick] (8.15, 0.05) -- (8.15, 1.45);
           \draw[thick] (8.05, 0) -- (8.25,0);
           \draw[thick] (8.05,1.5) -- (8.25,1.5);
            \draw (8.15,0.75) node[anchor=west] {$D$} ;

            \draw[thick] (0.05,-0.15) -- (5.2,-0.15);
           \draw[thick] (0,-0.05) -- (0,-0.25);
           \draw[thick] (5.25,-0.05) -- (5.25,-0.25);
            \draw (2.625, -0.4) node {$2nB^3$};
            \draw[thick] (5.3,-0.15) -- (7.95,-0.15);
           \draw[thick] (5.25,-0.05) -- (5.25,-0.25);
           \draw[thick] (8,-0.05) -- (8,-0.25);
            \draw (6.625, -0.4) node {$n(B^3+B)$};

		\end{tikzpicture}}
		\caption{Geometric Packing of the squares defined in the reduction, in the case where the 3-Partition instance admits a Partition. A zoomed-in depiction of the region where Enforcer and Partition tasks go can be found at the left, where each Partition block represents a pile of three Partition tasks whose side lengths sum up to $D$. Furthermore, in the context of the Demand Allocation problem, if there exists a feasible allocation of the tasks into a bin where the structure task of side length $nD$ starts at the leftmost time slot and the structure tasks of side length $(x^*+1)D$ are allocated to its right, this is the only possible allocation.}
		\label{fig:leung-packing}
	\end{figure}
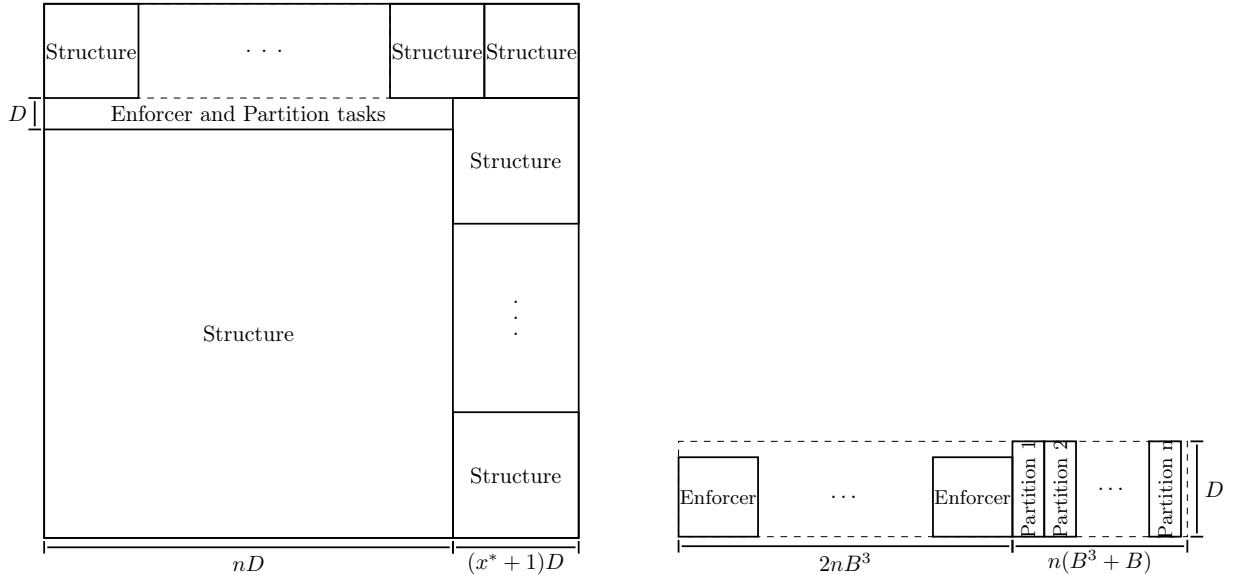 

We will now prove that, if the Demand Allocation instance admits a feasible allocation into a bin, then the 3-Partition instance admits a partition.

We will distinguish cases depending on where the structure task of side length $nD$ is placed. We will assume w.l.o.g. that this task starts in one of the first $D$ time slots. If it ends in one of the last $D$ time slots, we can simply reflect the allocation with respect to time slot $\lceil\frac{T}{2}\rceil$; if it starts after the first $D$ time slots, and ends before the last $D$ time slots, then every other structure task overlaps with this task, implying that the total area of other structure tasks in the allocation is at most $x^*(x^*+1)(x^*+2)D^2$, which is not possible as their total area is ${x^*}^2(x^*+2)D^2 + x^*(x^*+1)^2D^2$ which is strictly larger.

Suppose first that the structure task of side length $nD$ does not start in the leftmost time slot. This means that every structure task of side length $(x^*+1)D$ overlaps with this task and hence, since the total height of such two tasks would be $nD+(x^*+1)D = N$, they are all allocated into disjoint intervals of time slots. Furthermore, if there were some structure task of side length $(x^*+1)D$ using the time slot $x^*(x^*+1)D +1$, every structure task of side length $x^*D$ would overlap with it, which is not possible as they all must use that time slot and their total height is exactly $N$. Consequently, the structure tasks of side length $(x^*+1)D$ must start at the leftmost time slot of the timeline and must be contiguous, while the structure tasks of side length $x^*D$ must form a pile to the right of the structure task of side length $nD$ (see Figure~\ref{fig:structure_1_leung}).

We will now prove that the tasks can be rearranged so that every Enforcer and Partition task is allocated into the time slots $\{nD+1,\dots,(n+1)D\}$. Indeed, since structure tasks of side length $(x^*+1)D$ are allocated starting from the leftmost time slot and the structure tasks of side length $x^*D$ fill completely the time slot $(n+1)D+1$, Partition and Enforcer tasks can be decomposed into three disjoint (possibly empty) allocations: the ones to the left of the structure task of side length $nD$, the ones between the structure task of side length $nD$ and the the pile of structure tasks of length $x^*D$, and the ones to the right of the pile of structure tasks of side length $x^*D$ (see Figure~\ref{fig:structure_1_leung}). Suppose that the structure tasks of side length $x^*D$ are drawn geometrically as squares (i.e. without slicing them; since they form a pile this is possible, but Enforcer and Partition tasks may not be geometrically embedded after that), that the structure task of side length $nD$ is also drawn geometrically as a square at the bottom of the timeline, and that the structure tasks of side length $(x^*+1)D$ are drawn geometrically as squares on top of the structure task of side length $nD$ (which is again possible as there are at most two such task per time slot). Enforcer and Partition tasks are then allocated into three geometric regions defined by the boundary of the outside region, the boundary of the structure task of side length $nD$, the boundary of the structure tasks of side length $(x^*+1)D$, and the boundary of the pile of structure tasks of side length $x^*D$. Furthermore, by construction, these three regions fit together into a rectangle of width $D$ and height $nD$, meaning that the structure task of side length $nD$ can be shifted to the left so as to start at the leftmost time slot, the structure tasks of side length $x^*D$ can be shifted to the right so as to finish at the rightmost time slot, and the three regions (with their respecting tasks allocated) can be placed into the free rectangular area of width $D$ and height $nD$ at the time slots $\{nD+1,\dots,(n+1)D\}$ (see Figure~\ref{fig:structure_1_leung}).

 	\begin{figure}
	    \centering
    
		\resizebox{0.46\textwidth}{!}{\begin{tikzpicture}
            \fill[color=gray!30] (0,0) rectangle (0.2,7);
            \fill[pattern=north east lines] (7.2,0) rectangle (7.3,3);
            \fill[pattern=north east lines] (7.3,0) rectangle (7.4,1.5);
            \fill[pattern=north west lines] (8.9,0) rectangle (9,6);
            \fill[pattern=north west lines] (8.8,1.5) rectangle (8.9,4.5);
            \fill[pattern=north west lines] (8.7,3) rectangle (8.8,4.5);
            \fill[pattern=north east lines] (7.2,6) rectangle (7.5,7);
            \fill[pattern=north east lines] (7.2,4.5) rectangle (7.4,6);
            
            %\draw[dashed] (3,7) rectangle (7,8.5);
		  \draw[thick] (0,0) rectangle (9,9);
		  \draw[thick] (0.2,0) rectangle (7.2,7);
		  \draw[thick] (0,7) rectangle (2,9);
            \draw[thick] (5.5,7) rectangle (7.5,9);
		   \draw[thick] (9,9) rectangle (7.5,7.5);
		    \draw[thick] (7.5,6) rectangle (9,7.5);
		   \draw[thick] (7.4,0) rectangle (8.9,1.5);
            \draw[thick] (7.3,1.5) rectangle (8.8,3);
            \draw[thick] (7.2,3) rectangle (8.7,4.5);
            \draw[thick] (7.4,4.5) rectangle (8.9,6);
            
              \draw (1, 8) node {Structure};
              \draw (6.5, 8) node {Structure};
              %\draw (7.6, 4) node {$.$};
              %\draw (7.6, 4.25) node {$.$};
              %\draw (7.6, 4.5) node {$.$};
              \draw (8.15, 0.75) node {Structure};
              \draw (8.05, 2.25) node {Structure};
              \draw (8.25, 8.25) node {Structure};
              \draw (4, 8) node {$.$};
              \draw (3.5, 8) node {$.$};
              \draw (3.75, 8) node {$.$};
  		 \draw (3.45, 3.25) node {Structure};
              \draw (8.25, 6.75) node {Structure};
              \draw (7.95, 3.75) node {Structure};
              \draw (8.15, 5.25) node {Structure};
  			 %\draw (5.25, 6.75) node {Enforcer and Partition tasks};

            %\draw[thick] (8.65, 6.55) -- (8.65, 6.95);
           %\draw[thick] (8.55, 6.5) -- (8.75,6.5);
           %\draw[thick] (8.55,7) -- (8.75,7);
            %\draw (8.65,6.75) node[anchor=west] {C} ;
            
 %\draw[thick] (2.05,-0.15) -- (8.45,-0.15);
           %\draw[thick] (2,-0.05) -- (2,-0.25);
           %\draw[thick] (8.5,-0.05) -- (8.5,-0.25);
            %\draw (5.25, -0.4) node {$nC$};

        %\draw[thick] (0.05,-0.15) -- (1.95,-0.15);
           %\draw[thick] (2,-0.05) -- (2,-0.25);
           %\draw[thick] (0,-0.05) -- (0,-0.25);
            %\draw (1, -0.4) node {$(x^*+1)C$};

		\end{tikzpicture}}\hspace{20pt}
        \resizebox{0.17\textwidth}{!}{\begin{tikzpicture}
            \fill[color=gray!30] (0,0) rectangle (0.4,9.5);
            \fill[pattern=north east lines] (0.4,0) rectangle (0.6,4);
            \fill[pattern=north east lines] (0.6,0) rectangle (0.8,2);
            \fill[pattern=north west lines] (0.8,0) rectangle (1,8);
            \fill[pattern=north west lines] (0.6,2) rectangle (0.8,6);
            \fill[pattern=north west lines] (0.4,4) rectangle (0.6,6);
            \fill[pattern=north east lines] (0.4,6) rectangle (0.8,9.5);
            \fill[pattern=north east lines] (0.8,8) rectangle (1,9.5);
            %\draw[dashed] (3,9) -- (3,0) -- (0,0) -- (0,8) -- (1,8) -- (1,9);
            %\draw[dashed] (1,0) -- (1,8);
            \draw[thick] (3,10.5) -- (3,0) -- (-0.5,0);
		  \draw[thick] (1,8) rectangle (3,10);
            \draw[thick] (-0.5,9.5) -- (1,9.5) -- (1,10.5);
            \draw[thick] (0,0) -- (0,9.5);
            %\draw[thick] (1,4) -- (1,6.5);
            \draw (2,9) node {\large Structure} ;
		  %\draw[thick] (0.5,5.5) rectangle (2,7);
            %\draw (1.25,6.25) node {Structure} ;
            \draw[thick] (1,0) rectangle (3,2);
            \draw (2,1) node {\large Structure} ;
		  \draw[thick] (1,2) rectangle (3,4);
            \draw (2,3) node {\large Structure} ;
            \draw[thick] (1,4) rectangle (3,6);
            \draw (2,5) node {\large Structure} ;
		  \draw[thick] (1,6) rectangle (3,8);
            \draw (2,7) node {\large Structure} ;
            %\draw (2, 5) node {$.$};
            %\draw (2, 5.25) node {$.$};
            %\draw (2, 5.5) node {$.$};

            \draw[thick] (0.4,0) -- (0.4,9.5);
            \draw[thick] (0.8,0) -- (0.8,2) -- (0.6,2) -- (0.6,4) -- (0.4,4) -- (0.4,6) -- (0.8,6) -- (0.8,8) -- (1,8) -- (1,9.5);
            %\draw (0.6,5.35) node {\Large $\vdots$};
            %\draw[thick] (1,6.5) -- (0.4,6.5) -- (0.4,5.5);

            %\draw[thick] (5.25,0) rectangle (5.75,1.5);
            %\draw (5.5,0.75) node[rotate=90] {\small Partition 1} ;
            %\draw[thick] (5.75,0) rectangle (6.25,1.5);
            %\draw (6,0.75) node[rotate=90] {\small Partition 2} ;
            %\draw[thick] (7.4,0) rectangle (7.9,1.5);
            %\draw (7.65,0.75) node[rotate=90] {\small Partition n} ;
            %\draw (6.825,0.75) node {$\dots$};

            %\draw[thick] (8.15, 0.05) -- (8.15, 1.45);
           %\draw[thick] (8.05, 0) -- (8.25,0);
           %\draw[thick] (8.05,1.5) -- (8.25,1.5);
            %\draw (8.15,0.75) node[anchor=west] {C} ;

            %\draw[thick] (0.05,-0.15) -- (5.2,-0.15);
           %\draw[thick] (0,-0.05) -- (0,-0.25);
           %\draw[thick] (5.25,-0.05) -- (5.25,-0.25);
            %\draw (2.625, -0.4) node {$2nB^3$};
            %\draw[thick] (5.3,-0.15) -- (7.95,-0.15);
           %\draw[thick] (5.25,-0.05) -- (5.25,-0.25);
           %\draw[thick] (8,-0.05) -- (8,-0.25);
            %\draw (6.625, -0.4) node {$n(B^3+B)$};

		\end{tikzpicture}}\hspace{40pt}
        \resizebox{0.185\textwidth}{!}{\begin{tikzpicture}
            \draw[dashed] (0,0) rectangle (2,10);

            \draw[thick] (0.1,0) rectangle (1.9,1.8);
            \draw[thick] (0,1.8) rectangle (1.8,3.6);
            \draw[thick] (0.2,4.6) rectangle (2,6.4);
            \draw (1,4.2) node {$\vdots$};

            \draw[thick] (0,6.4) rectangle (0.7,7.1);
            \draw[thick] (0.7,6.4) rectangle (1.35,7.05);
            \draw[thick] (1.35,6.4) rectangle (2,7.05);
            \draw[thick] (0,7.1) rectangle (0.66,7.76);
            \draw[thick] (0.66,7.1) rectangle (1.34,7.78);
            \draw[thick] (1.34,7.1) rectangle (2,7.76);
            \draw[thick] (0,9.2) rectangle (0.67,9.87);
            \draw[thick] (0.67,9.2) rectangle (1.35,9.88);
            \draw[thick] (1.35,9.2) rectangle (2,9.85);
            \draw (1,8.56) node {$\vdots$};

            \draw (0.33,6.76) node[rotate=45] {\scriptsize $B^3+a_i$};
            \draw (1,6.75) node[rotate=45] {\scriptsize $B^3+a_j$};
            \draw (1.65,6.75) node[rotate=45] {\scriptsize $B^3+a_k$};
            \draw (1,0.9) node {Enforcer};
            \draw (0.9,2.7) node {Enforcer};
            \draw (1.1,5.5) node {Enforcer};

            \draw[thick] (2.15, 0.05) -- (2.15, 6.35);
            \draw[thick] (2.05, 0) -- (2.25,0);
            \draw[thick] (2.05,6.4) -- (2.25,6.4);
            \draw (2.15,3.2) node[anchor=west] {\small $2nB^3$} ;

            \draw[thick] (2.15, 6.45) -- (2.15, 9.95);
            \draw[thick] (2.05, 6.4) -- (2.25,6.4);
            \draw[thick] (2.05,10) -- (2.25,10);
            \draw (2.15,8.2) node[anchor=west] {\small $nB^3+B$} ;

		\end{tikzpicture}}
		\caption{The case when the structure task of side length $nD$ does not start at the leftmost time slot. In the middle figure, structure tasks of side length $x^*D$ are pushed to the right and the structure task of side length $nD$ is pushed to the left, so that Enforcer and Partition tasks are all allocated into the colored rectangular region of height $nD$ and width $D$. In the right, a depiction of how Enforcer and Partition tasks must be allocated, provided that they all fit.}
		\label{fig:structure_1_leung}
	\end{figure}
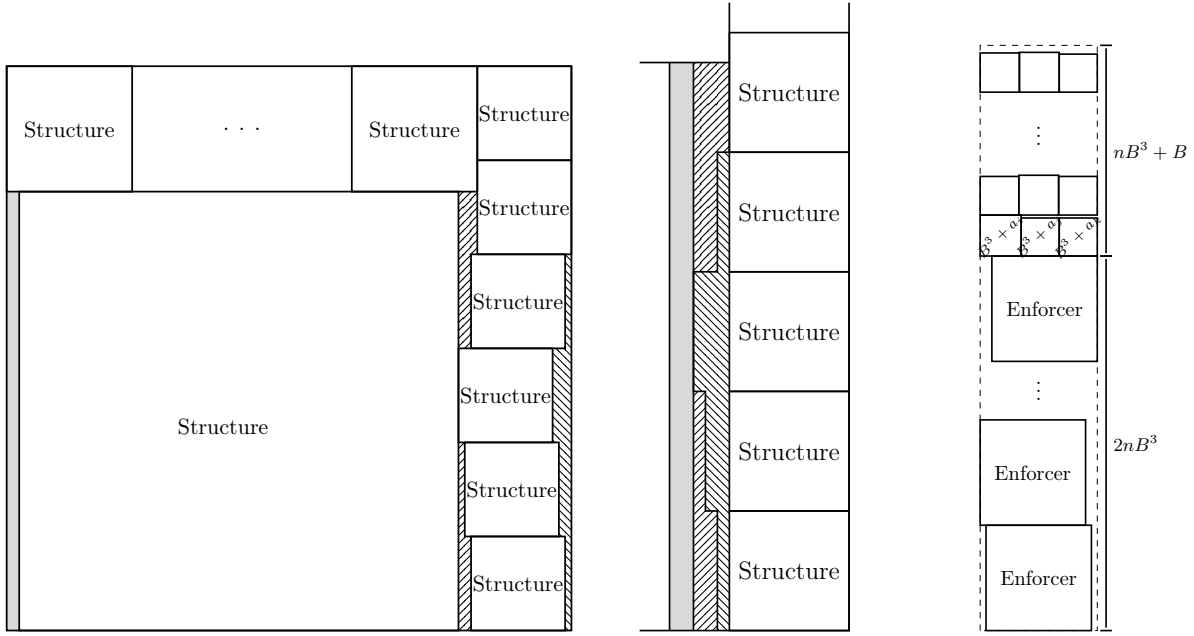 

We will look now the aforementioned rectangular region where Enforcer and Partition tasks are placed. Observe first that, since Enforcer tasks have side length $3B^3$ and the region has width $3B^3+B$, all of them must form a pile of total height $2nB^3$. Since the Partition tasks have width at least $B^3$, all of them must overlap with the aforementioned pile, implying that in each time slot there must be at most $n$ Partition tasks; if there were more, their total height together with the Enforcer tasks would be at least $(3n+1)B^3 > 3nB^3+nB = nD$. As the total width of the Partition tasks is $3nB^3 + nB = nD$ (i.e. $n$ times the width of the region), there must exist a way to partition them into groups of three so that the total width of each group is exactly $D$; if not, for any such partition into groups of three there would a group whose total width is larger that $D$, hence not fitting into the region and forcing to put at least $n+1$ tasks into some time slot, which is not possible. If a group of three Partition tasks has total width exactly $3B^3+B$, this means that the corresponding numbers from the $3$-Partition instance sum up to exactly $B$, implying that there is a $3$-Partition as desired (see Figure~\ref{fig:structure_1_leung}).

Consider now the case when the structure task of side length $nD$ is allocated at the leftmost time slot. We will partition the remaining structure tasks into two sets: $S_1$ are the tasks that overlap with the structure task of side length $nD$, and $S_2$ are the tasks that do not; the total side length of the tasks in both sets is $x^*(x^*+2)D+x^*(x^*+1)D$. We will prove that the total side length of the tasks in $S_1$ must be exactly $x^*(x^*+1)D$: It cannot be smaller that $x^*(x^*+1)D$ because the total side length of tasks in $S_2$ is at most $N = x^*(x^*+2)D$ as they all occupy the time slot $x^*(x^*+1)D+1$, and it cannot be larger because, in that case, there would some task from $S_1$ occupying the time slot $x^*(x^*+1)D+1$. This implies that the total side length of $S_2$ is at most $N-x^*D = x^*(x^*+1)D$, which is not possible as the total side length of $S_1$ is strictly smaller than $N$ as no task can occupy the time slot $x^*(x^*+1)D$ and the rightmost time slot simultaneously.

We will now use the following fact stated by Leung et al.~\cite{leung}: The equation $k_1 x^*D + k_2 (x^*+1)D = x^*(x^*+1)D$ has only two solutions: Either $k_1 = x^*+1$ and $k_2=0$, or $k+1 = 0$ and $k_2 = x^*$. In the context of the previous solution, this means that either $S_1$ corresponds to $x^*$ tasks of side length $x^*+1$, or to $x^*+1$ tasks of side length $x^*$. The first case is the same as the one already considered below, so we will focus on the second case (see Figure~\ref{fig:leung-packing}). 

If we draw the structure tasks of side length $x^*D$ geometrically as squares touching the upper boundary of the region, it can be seen that Enforcer and Partition tasks are allocated into a rectangular region of width $nD$ and height $D$. Similarly to the previous case, since Enforcer tasks have side length $3B^3$ and the region has width $3B^3+B$, all of them use in total $2nB^3$ time slots. Since the Partition tasks have side length at least $B^3$, it is not possible to place more than $n$ Partition tasks into disjoint intervals of time slots. As the total side length of the Partition tasks is $3nB^3 + nB = nD$ (i.e. $n$ times the height of the region), there must exist a way to partition them into groups of three so that the total side length of each group is exactly $D$; if not, for any such partition into groups of three there would a group whose total side length is larger that $D$, hence not fitting into the region and forcing to put at least $n+1$ tasks into disjoint intervals of time slots, which is not possible. If a group of three Partition tasks has total width exactly $3B^3+B$, this means that the corresponding numbers from the $3$-Partition instance sum up to exactly $B$, implying that there is a $3$-Partition as desired. This concludes the proof of Lemma~\ref{lem:hardness-leung}.

    \begin{corollary}
 For any $\varepsilon > 0$, there is no $(2-\varepsilon)$-approximation algorithm for 2D Demand Bin Packing for square tasks and square bins, unless P=NP.
 \end{corollary} 
  \begin{proof}
Let us assume that we have a $(2-\varepsilon)$-approximation algorithm for 2D Demand Bin Packing for square tasks. Given an instance of the Demand Allocation problem restricted to squares tasks and square bins, we can decide whether $OPT=1$ by applying this algorithm since a $(2-\varepsilon)$-approximation returns the optimal solution when $OPT=1$. This means that we can solve Demand Allocation restricted to square tasks and square bins in polynomial time, a contradiction to Lemma~\ref{lem:hardness-leung}  unless P=NP.
 \end{proof}
 
\section{2D Demand Bin Packing for Square Tasks vs 2D Geometric Bin Packing for Square Tasks}\label{app:gap_instance}

In this section, we will prove that there exist instances of 2D Demand Bin Packing for square tasks that can be allocated into a single square bin but, when considered as an instance of 2D Geometric Bin Packing, more than a bin is required. This implies that just applying an algorithm for 2D Geometric Bin Packing to the 2D Demand Bin Packing instance might not lead to tight results, even in the presence of square tasks and square bins. The construction is a slight adaptation of the construction due to Gálvez et al.~\cite{GGJK21}.

\begin{lemma}\label{lem:gap_demand}
There exist instances of 2D Demand Bin Packing with square tasks that can be allocated into a single bin, but there is no geometric packing of the tasks into the bin.
 \end{lemma}
\begin{proof}

Consider the following $2D$ Demand Bin Packing instance:

\begin{itemize}
    \item The bins are defined by a timeline of length $T=21$ and capacity $C=21$.
    \item There are two square tasks, denoted by $i_1$ and $i_2$, of side length $10$.
    \item There are two square tasks, denoted by $i_3$ and $i_4$, of side length $8$.
    \item There is one square task, denoted by $i_5$, of side length $5$.
    \item There are nine square tasks, denoted by $i_6,\dots,i_{14}$, of side length $3$.
\end{itemize}

Figure~\ref{fig:gapinstance} shows a way to allocate all the tasks into a single bin. We will now prove that there is no geometric packing of these squares into a square region of side length $21$. 

For the sake of contradiction, suppose that there exists a geometric packing of the squares into the region. We assume that the coordinates of the squares in the region are all integral, simply by pushing them  to the bottom and to the left as much as possible; they either touch the boundary of the region, or the boundary of some other square, which leads to integral coordinates as the side lengths are all integral. We observe first that, since the area of the region is $441$ and the total area of the squares is $434$, only a total area of $7$ can be free in the region. This implies that squares $i_1, i_2, i_3$ and $i_4$ cannot be placed at distance $1$ or $2$ from any boundary, as this would imply that the area between the square and the boundary of the region is free (as every remaining square has side length at least $3$) and larger than $7$. Hence, every such square is either at the boundary of the region, or at distance at least three from every boundary.

We will distinguish two cases depending on the positions of $i_1$ and $i_2$ in the packing:

\begin{itemize}
    \item Suppose that any vertical or horizontal line that crosses the region completely does not intersect the interiors of both $i_1$ and $i_2$ simultaneously.  Since these squares cannot be placed at distance $1$ or $2$ from a boundary, this implies that $i_1$ and $i_2$ are placed into opposite corners of the region (see Figure~\ref{fig:gap_structure}). Such a placement defines two disjoint regions where the remaining squares must be placed: a square region of side length $11$, which we denote by $A$, and the complement of the union of $A$, $i_1$ and $i_2$ denoted by $A'$. 

    The area of $A$ is $121$ and the area of $A'$ is $120$; this means that $i_3$ and $i_4$, since each square has area $64$, cannot be placed into the same region. This leads to a contradiction because $i_5$ does not fit in any of these regions together with a square of side length $8$.

    	\begin{figure}
	    \centering
    
		\scalebox{.6}{\begin{tikzpicture}      
            \draw[thick] (0,0) rectangle (10.5,10.5);
  		 \draw[fill=lightgray!15] (0,0) rectangle (5,5);
  		 \draw[fill=lightgray!15] (5.5,5.5) rectangle (10.5,10.5);
  		 \draw (2.5, 2.5) node {$\mathbf{i_1}$};
  		 \draw (8, 8) node {$\mathbf{i_2}$};
  		 
  		  \draw[thick] (-0.1, 10.5) -- (-0.5,10.5);
  		  \draw[thick] (-0.3, 10.45) -- (-0.3,5.05);
  		  \draw[thick] (-0.1, 5) -- (-0.5,5);
  		   \draw (-0.3, 7.75) node[anchor=east] {$11$};

  		  \draw[thick] (0,5) rectangle (5.5,10.5);
  		  \fill[pattern = north east lines, pattern color = black!60] (0,5) rectangle (5.5,10.5);
          \draw (2.75, 7.75) node {$\mathbf{A}$};
          \draw (7.75, 2.75) node {$\mathbf{A'}$};
  		   
		\end{tikzpicture}}\hspace{13pt}
        \scalebox{.6}{\begin{tikzpicture}
            
		  \draw[thick] (0,0) rectangle (10.5,10.5);
  		 \draw[fill=lightgray!15] (0,2) rectangle (5,7);
  		 \draw[fill=lightgray!15] (5.5,4) rectangle (10.5,9);
  		 \draw (2.5, 4.5) node {$\mathbf{i_1}$};
  		 \draw (8, 6.5) node {$\mathbf{i_2}$};
  		 
  		  \draw[thick] (5.5, 3.90) -- (5.5,2.05);
  		  \draw[thick] (5.3, 3.95) -- (5.7,3.95);
  		  \draw[thick] (5.3, 2) -- (5.7,2);
  		 \draw (5.9, 3) node {$ \ge 3$};
  		 
          \draw[thick] (-0.1, 10.5) -- (-0.5,10.5);
  		  \draw[thick] (-0.3, 10.45) -- (-0.3,7.05);
  		  \draw[thick] (-0.1, 7) -- (-0.5,7);
  		   \draw (-0.6, 8.75) node {$\ell_{top}$};
  		   
          \draw[thick] (-0.1, 2) -- (-0.5,2);
  		  \draw[thick] (-0.3, 1.95) -- (-0.3,0.05);
  		  \draw[thick] (-0.1, 0) -- (-0.5,0);
  		   \draw (-0.6, 1) node {$\ell_{bot}$};
  		   
          \draw[thick] (10.6, 4) -- (11,4);
  		  \draw[thick] (10.8, 3.95) -- (10.8,0.05);
  		  \draw[thick] (10.6, 0) -- (11,0);
  		   \draw (11.1, 2) node {$r_{bot}$};
           \draw[thick] (10.6, 10.5) -- (11,10.5);
  		  \draw[thick] (10.8, 10.45) -- (10.8,9.05);
  		  \draw[thick] (10.6, 9) -- (11,9);
  		   \draw (11.1, 9.75) node {$r_{top}$};
  		   
		\end{tikzpicture}}
		\caption{Case distinction in the proof of Lemma~\ref{lem:gap_demand}. Left: when no vertical nor horizontal line intersect the interior of $i_1$ and $i_2$ simultaneously. Right: When there exists an horizontal line intersecting the interior of $i_1$ and $i_2$ simultaneously.}
	    \label{fig:gap_structure}
	\end{figure}
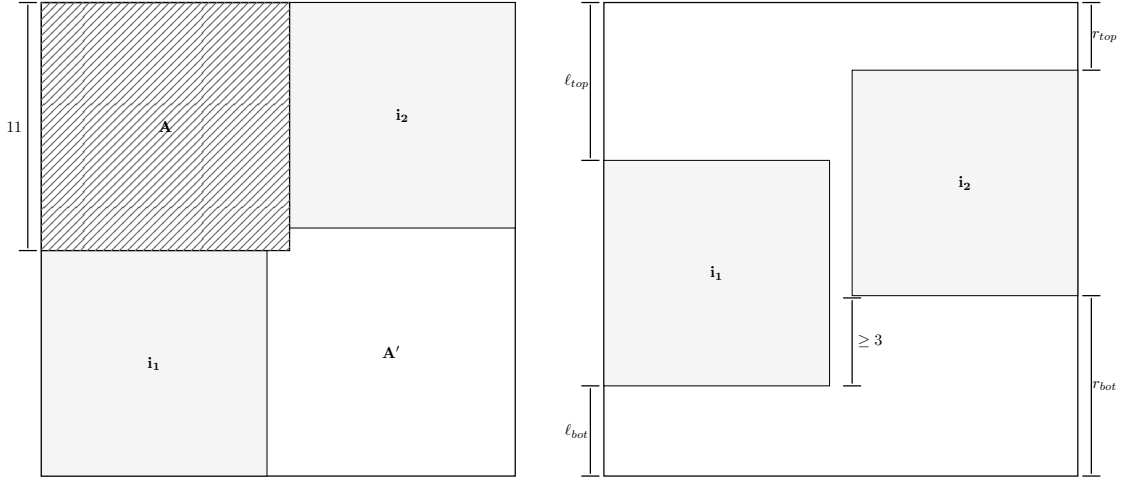

\item Suppose instead that there exists some horizontal line (the case of a vertical line is analogous as we can rotate the packing by $90$ degrees) that intersects the interior of $i_1$ and $i_2$ simultaneously. 

Since the total side length of $i_1$ and $i_2$ is $20$ and the side length of the region is $21$, each square must be placed at a vertical boundary of the region. Furthermore, the rectangular region formed between them must have height at most $7$ because of the fact that this space will remain free since it has width $1$. Thus one of the squares should be at least 3 units higher (or lower) than the other. Without loss of generality, we assume that $i_1$ is packed along the left boundary and $i_2$ at the right boundary, at a higher position than $i_1$ (see Figure~\ref{fig:gap_structure}).

	Let $\ell_{top}$ be the remaining height above $i_1$, $\ell_{bot}$ the remaining height below $i_1$, $r_{top}$ the remaining height above $i_2$, and $r_{bot}$ the remaining height below $i_2$, respectively. The placement of the squares partitions the bin into two disjoint regions, one below $i_1$ and $i_2$, and one above $i_1$ and $i_2$, where the remaining squares must be placed. In order to pack $i_3$ and $i_4$, we need to ensure that either $\ell_{top} \ge 8$ or $\ell_{bot} \ge 8 $, and either $r_{top} \ge 8$ or $r_{bot} \ge 8 $ since the side length of these squares is $8$. The case $\ell_{bot} \ge 8 $ and $r_{top} \ge 8$ cannot happen because $i_2$ must be at a higher position with respect to $i_1$. We will thus look at the remaining cases. 
	
	If $\ell_{top},r_{bot} \ge 8$ (see Figure~\ref{fig:gapcases}), this means that $\ell_{bot}, r_{top}\le 3$, and hence it is not possible to place $i_3$ and $i_4$ in the same region (i.e. either above or below $i_1$ and $i_2$). However, $i_5$ also does not fit below $i_1$ or above $i_2$, and it cannot be placed together with a square of side length $8$ in the remaining area, implying that this case cannot occur.

	If $\ell_{top},r_{top} \ge 8$ (see Figure~\ref{fig:gapcases}), then $\ell_{bot}=0$ and $r_{top}=8$. This implies that the rectangular region between $i_1$ and $i_2$ has height $7$ and width $1$, so every other space in the region must be occupied by some square. However, the rectangular region below $i_2$ has height $3$ and width $11$, and hence cannot be full with squares, as only squares of side length $3$ fit in there, but $11$ is not divisible by $3$. This again leads to a contradiction. The case $\ell_{bot},r_{bot}\ge 8$ is analogous, as we can rotate the packing by $180$ degrees and swap $i_1$ and $i_2$.

    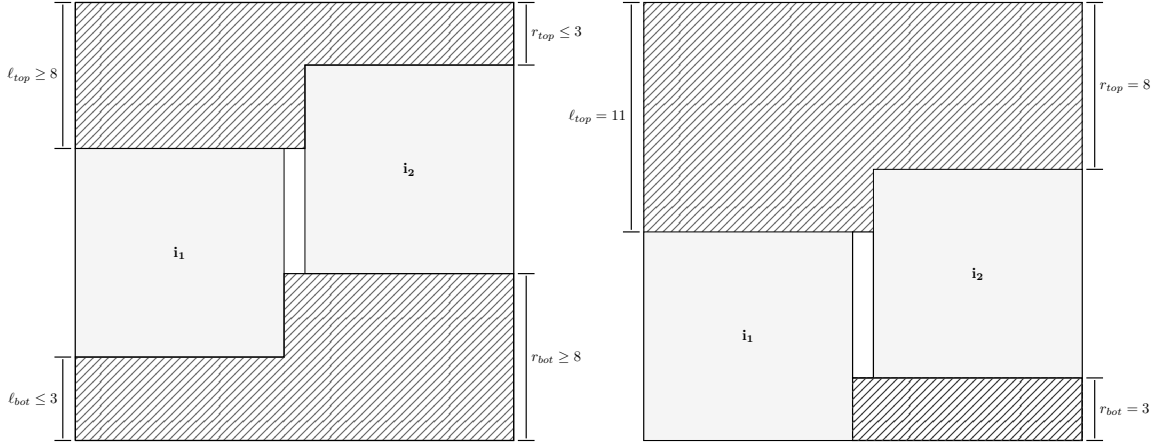
\begin{figure}
	    \centering
    
		\scalebox{.555}{\begin{tikzpicture}      
            
		  \draw[thick] (0,0) rectangle (10.5,10.5);
  		 \draw[fill=lightgray!15] (0,2) rectangle (5,7);
  		 \draw[fill=lightgray!15] (5.5,4) rectangle (10.5,9);
  		 \draw (2.5, 4.5) node {$\mathbf{i_1}$};
  		 \draw (8, 6.5) node {$\mathbf{i_2}$};
  		   		  
  		   %\draw[thick] (0,7) rectangle (5.5,10.5);
  		   %\draw[thick] (5,0) rectangle (10.5,4);
  		   \fill[pattern = north east lines, pattern color = black!60] (0,7) rectangle (5.5,10.5);
              \fill[pattern = north east lines, pattern color = black!60] (10.5,9) rectangle (5.5,10.5);
  		   \fill[pattern = north east lines, pattern color = black!60] (5,0) rectangle (10.5,4);
              \fill[pattern = north east lines, pattern color = black!60] (0,0) rectangle (5,2);
              \draw[thick] (0,0) -- (10.5,0) -- (10.5,4) -- (5,4) -- (5,2) -- (0,2) -- (0,0);
              \draw[thick] (0,10.5)  -- (10.5,10.5) -- (10.5,9) -- (5.5,9) -- (5.5,7) -- (0,7) -- (0,10.5);
  		   %\draw (2.75, 8.75) node {$A_1$};
  		   %\draw (7.75, 2) node {$A_2$};

          \draw[thick] (-0.1, 10.5) -- (-0.5,10.5);
  		  \draw[thick] (-0.3, 10.45) -- (-0.3,7.05);
  		  \draw[thick] (-0.1, 7) -- (-0.5,7);
  		   \draw (-0.3, 8.75) node[anchor=east] {$\ell_{top} \ge 8$};
  		   
          \draw[thick] (-0.1, 2) -- (-0.5,2);
  		  \draw[thick] (-0.3, 1.95) -- (-0.3,0.05);
  		  \draw[thick] (-0.1, 0) -- (-0.5,0);
  		   \draw (-0.3, 1) node[anchor=east] {$\ell_{bot} \le 3$};
  		   
          \draw[thick] (10.6, 4) -- (11,4);
  		  \draw[thick] (10.8, 3.95) -- (10.8,0.05);
  		  \draw[thick] (10.6, 0) -- (11,0);
  		   \draw (10.8, 2) node[anchor=west] {$r_{bot} \ge 8$};
           \draw[thick] (10.6, 10.5) -- (11,10.5);
  		  \draw[thick] (10.8, 10.45) -- (10.8,9.05);
  		  \draw[thick] (10.6, 9) -- (11,9);
  		   \draw (10.8, 9.75) node[anchor=west] {$r_{top} \le 3$};
  		   
		\end{tikzpicture}}\hspace{-12pt}
		\scalebox{.555}{\begin{tikzpicture}
			
            \fill[pattern = north east lines, pattern color = black!60] (0,5) rectangle (5.5,10.5);
              \fill[pattern = north east lines, pattern color = black!60] (10.5,6.5) rectangle (5.5,10.5);
            
		  \draw[thick] (0,0) rectangle (10.5,10.5);
  		 \draw[fill=lightgray!15] (0,0) rectangle (5,5);
  		 \draw[fill=lightgray!15] (5.5,1.5) rectangle (10.5,6.5);
  		 \draw (2.5, 2.5) node {$\mathbf{i_1}$};
  		 \draw (8, 4) node {$\mathbf{i_2}$};
  		 		 
          \draw[thick] (-0.1, 10.5) -- (-0.5,10.5);
  		  \draw[thick] (-0.3, 10.45) -- (-0.3,5.05);
  		  \draw[thick] (-0.1, 5) -- (-0.5,5);
  		   \draw (-0.3, 7.75) node[anchor=east] {$\ell_{top} = 11$};

          \draw[thick] (10.6, 1.5) -- (11,1.5);
  		  \draw[thick] (10.8, 1.45) -- (10.8,0.05);
  		  \draw[thick] (10.6, 0) -- (11,0);
  		   \draw (10.8, 0.75) node[anchor=west] {$r_{bot} = 3$};
           \draw[thick] (10.6, 10.5) -- (11,10.5);
  		  \draw[thick] (10.8, 10.45) -- (10.8,6.55);
  		  \draw[thick] (10.6, 6.5) -- (11,6.5);
  		   \draw (10.8, 8.5) node[anchor=west] {$r_{top} = 8$};
  		   
  		   \draw[thick] (5,0) rectangle (10.5,1.5);
  		    %\draw (7.75, 0.75) node {$A_3$};
  		   \draw[thick] (5,1.5) rectangle (5.5,5);
            %\fill[color=lightgray] (5,1.5) rectangle (5.5,5);
           \fill[pattern = north east lines, pattern color = black!80] (5,0) rectangle (10.5,1.5);
  		   
		\end{tikzpicture}}
		\caption{Cases when there is a horizontal line intersecting the interior of $i_1$ and $i_2$ simultaneously. Left: The case $\ell_{top},r_{bot}\ge 8$. Right: The case $\ell_{top},r_{top}\ge 8$.}
	    \label{fig:gapcases}
	\end{figure} 
\end{itemize}

Since all the cases lead to a contradiction, we conclude that there is no geometric packing of the squares into the region. \end{proof}

\section{A 2-approximation for 2D Demand Bin Packing for Square tasks}\label{app:2apxsquare}

In this section, a $2$-approximation algorithm for 2D Demand Bin Packing for square tasks is presented. We will start by proving how to compute a $2$-structured solution for tasks of height larger than $\alpha H$ or width larger than $\beta W$, where $\alpha$ and $\beta$ are constants such that $(1-\alpha) (1-\beta) \ge \frac{1}{2}$. This solution will allow us to include the remaining tasks with the help of Lemma~\ref{lem:smalltasks} and conclude the result. Throughout this section, we make us of Next-Fit Decreasing Height algorithm (NFDH) which is defined in the following:

Let us sort the tasks non-increasingly according to their heights (breaking ties arbitrarily). Then the algorithm works in rounds $j\ge 1$: At the beginning of round $j$, each incoming task is placed one next to the other, starting from the leftmost time slot, until a task does not fit. These tasks define a shelf, which is place as such in the current bin if it fits, and otherwise it is placed into a new empty bin (the previous bin is closed and not visited again. It is not difficult to see that the output of this algorithm consists of bins with sorted load profiles.

\begin{lemma}\label{lem:2epslem}
For any instance $I$, we can find $\alpha$ and $\beta$ with $(1-\alpha) (1-\beta) \ge \frac{1}{2}$ such that we can compute a $2$-structured solution for all the tasks of height larger than $\alpha C$ or width larger than $\beta T$.
\end{lemma}

\begin{proof}

We will consider different cases depending on the relation between $C$ and $T$. 

\begin{enumerate}
    \item Consider first the case $T\ge \frac{4}{3}C$. We will prove that there exists a well-structured feasible solution of the squares having height larger than $\frac{1}{3}C$, and then complete the solution via Lemma~\ref{lem:smalltasks}. Notice that the requirements of the lemma would be fulfilled as the tasks remaining to be allocated have height at most $\frac{1}{3}C$ and width at most $\frac{1}{3}C \le \frac{1}{4}T$ and it holds that $\left(1-\frac{1}{3}\right)\left(1-\frac{1}{4}\right)=\frac{1}{2}$.
    
    Similarly as in the proof of Lemma~\ref{lem:shortapx}, it is possible to place all the tasks of height larger than $\frac{1}{3}C$ except for one (let us call it $k$) into a rectangular region of height $C$ and width $OPT\cdot T$ (see Figure~\ref{fig:2epspacking}). Let $I_1$ be the set of tasks in the bottom row (including $k$) and $I_2$ be the tasks in the top row of the construction respectively. We will derive a well-structured feasible solution into at most $2OPT$ bins out of the construction as follows: We will first draw vertical lines at every integer multiple of $C$ so as to divide the region into $OPT$ rectangular regions $R_1, R_2, \dots, R_{OPT}$ of width $T$ and height $C$. Consider now the tasks intersecting some region $R_i$, we will place into one bin all tasks from $I_1$ completely contained in $R_i$ together with the task from $I_2$ intersected by the left boundary of $R_i$ (if any) using NFDH, and in another bin the task from $I_1$ intersected by the right boundary of $R_i$ together with the tasks from $I_2$ completely contained in $R_i$ using NFDH. It is not difficult to see that NFDH can pack the tasks into the corresponding bins actually using at most two shelves each time as even if the first shelf had only tasks from $I_1$ (in each corresponding case) the second shelf still fits on top. This way we obtain the required solution by iterating the previous procedure on all the bins from left to right, noticing that every task was assigned to some bin as no task intersects the left boundary of $R_1$ and no task from $I_2$ intersects the right boundary of $R_{OPT}$.
    
    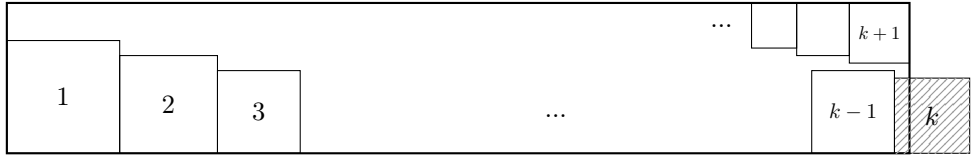
\begin{figure}
	    \centering
		\begin{tikzpicture}
		  \draw[thick] (0,0) rectangle (12,2);
  		  \draw (0,0) rectangle (1.5,1.5);
  		   \draw (1.5,0) rectangle (2.8,1.3);
  		    \draw (2.8,0) rectangle (3.9,1.1);
  		   \draw (11.8,0) rectangle (12.8,1);
  		  \fill[pattern = north east lines, pattern color = gray!80] (11.8,0) rectangle (12.8,1);
			\draw (11.8,1.1) rectangle (10.7,0);
            \draw (7.3, 0.5) node {$...$};
\draw (11.2,2) rectangle (12,1.2);
\draw (10.5,2) rectangle (11.2,1.3);
\draw (9.9,2) rectangle (10.5,1.4);
\draw (9.5, 1.7) node {$...$};
\draw (0.75, 0.75) node {1};
\draw (2.15, 0.65) node {2};
\draw (3.35, 0.55) node {3};
\draw (12.3, 0.5) node {$k$};
\draw (11.6, 1.6) node[scale=0.7] {$k+1$};
\draw (11.25, 0.55) node[scale=0.8] {$k-1$};
		\end{tikzpicture}
		\caption{Case 1 of Lemma~\ref{lem:2epslem}, when $T \ge \frac{4}{3}C$. A $2$-structured solution can be derived from the depicted packing.}
		\label{fig:2epspacking}
			\end{figure} 
    
    \item Consider now the case $C \le T < \frac{4}{3}C$. Our goal is to compute a well-structured feasible solution of the tasks having width larger than $\frac{1}{4}T$, as then the remaining tasks can be allocated using Lemma~\ref{lem:smalltasks} since their height will be at most $\frac{1}{4}T \le \frac{1}{3}V$ and their width will be at most $\frac{1}{4}T$, and it holds that $\left(1-\frac{1}{4}\right)\left(1-\frac{1}{3}\right) = \frac{1}{2}$.
    
    Let $\varepsilon=1/55$ and suppose that $OPT\le \frac{1}{\varepsilon}$. Since tasks are squares, in this case the number of tasks having width larger than $\frac{1}{4}T$ in the whole instance is a constant as at most nine of them can be allocated in each bin. Hence we can simply guess the assignment of these tasks into the bins optimally. If in this assignment a bin is $\frac{1}{2}$-full we leave it as it is, otherwise we will rearrange it using NFDH. We claim that his procedure returns an assignment into at most two bins: Indeed, if two bins are not sufficient, this means that we created shelves defined by tasks $i_1, i_2,...,i_k$ and some task $i_{k+1}$ does not fit neither in the last shelf nor on top of it in the second bin (notice that only $i_1$ can have width larger than $\frac{1}{2}T$). The total allocated area then would be at least \[h_{i_2}(T - w_{i_2}) + h_{i_3}(T - w_{i_3})+...+h_{i_{k+1}}(W - w_{i_{k+1}}) \ge (h_{i_2} + h_{i_3} +...+ h_{i_{k+1}}) \frac{1}{2}T > \frac{1}{2}TC\] which is a contradiction. This procedure then returns a $2$-structured solution for the tasks.
    
    If on the other hand $OPT\ge \frac{1}{\varepsilon}$, We will pack the tasks of width larger than $\frac{1}{3}T$ in a similar fashion as Theorem~\ref{thm:2apx-square-square}, and then place the tasks having width between $\frac{1}{4}T$ and $\frac{1}{3}T$ with the help of the APTAS for Bin Packing~\cite{VL81}. 
    
    More in detail, we will pack the tasks having width larger than $\frac{1}{3}T$ as follows: We initially sort the tasks non-increasingly by height, and for each task in this order we try to place it in any of the current bins as much to the left as possible; if the task does not fit in any bin, we open a new one and allocate the task there. Notice that in the optimal solution restricted to these tasks, each bin has at most four tasks, at most one of them having width larger than $\frac{1}{2}T$, and it is always possible to partition the slots into disjoint intervals such that each task occupies slots of only one side. 
    
    We will prove by induction on the number of tasks $k$ having width between $\frac{1}{3}T$ and $\frac{1}{2}T$ that this procedure is optimal. If $k=0$, we use the optimal number of bins as tasks having width larger than $\frac{1}{2}T$ must go to different bins. Now, for the inductive step, consider the list of tasks and remove the smallest task $j$. Let $OPT_{k+1}$ denote the optimal number of bins for the whole instance and $OPT_{k}$ for the instance without task $j$. By induction, for the remaining tasks we use $OPT_{k}$ bins. If task $j$ fits in one of the aforementioned intervals, we place it and get $OPT_{k} \le OPT_{k+1}$ bins, proving the claim. If it does not fit in any possible interval then no other task from the list fits in these regions; indeed, if an interval does not have tasks on it, no task from the list fits here because they are too wide. If an interval contains one task, no task from the list fits together with this task in any packing. Finally, if an interval contains two tasks it is not possible to fit any other task from the list in there. This implies that $OPT_{k+1} > OPT_{k}$ and hence opening a new bin is actually required.
			
	Notice that if a bin contains four tasks, either they all have width between $\frac{1}{3}T$ and $\frac{1}{2}T$ and we can make the load profile non-increasing, or one of them has width larger than $\frac{1}{2}T$ and then the total area of the tasks is at least $\frac{1}{4}TC + 3 \cdot \frac{1}{9}TC > \frac{1}{2}TC$. If a bin contains three tasks, either we can put one task on top of the tallest one at the left to achieve a non-increasing profile or its total area is larger than $\frac{1}{2}TC$ as the three tasks can cover the left half of the bin. If a bin contains one or two tasks, the load profile can always be made non-increasing. 
	
	For tasks having width between $\frac{1}{4}T$ and $\frac{1}{3}T$, we notice that in the optimal solution the slots of the bins can always be partitioned into three intervals of width $\frac{1}{3}T$ where tasks lie. Every task is intersected by at least one of the vertical lines at positions $\frac{1}{4}T, \frac{1}{2}T$ and $\frac{3}{4}T$. Tasks intersected by the first line can be moved to the leftmost interval such that it contains them completely, then among the remaining tasks we can move to the second intervals the ones intersected by the second line and similarly for the rest (see Figure~\ref{fig:2eps3reg}). These intervals induce then a Bin Packing instance of optimal value at most $3OPT$ for the tasks, which we can solve almost optimally using the APTAS for Bin Packing. However, we need to carefully bound the number of used bins in total (including the first phase for tasks of width larger than $\frac{1}{3}T$). 
	
	If the first phase ends up using at most $(1-2\varepsilon)OPT$ bins, then since the number of bins obtained from the APTAS is at most $(1+2\varepsilon)OPT$ we obtain the claimed solution. If on the other hand the first phase uses more than $(1-2\varepsilon)OPT$ bins, we will show then that the tasks having width between $\frac{1}{4}T$ and $\frac{1}{3}T$ can be arranged into $\left(\frac{8}{9}+2\varepsilon\right)OPT$ bins only, which would imply that by using the APTAS we obtain at most $OPT$ bins for these tasks (as $\varepsilon\le 1/55$) and consequently at most $2OPT$ bins in total as required. Indeed, if the first phase uses more than $(1-2\varepsilon)OPT$ bins, this means in particular that in the optimal solution at least $(1-2\varepsilon)OPT$ have tasks of width larger than $\frac{1}{3}T$. Consider nine such bins arbitrarily, we can remove all the tasks from $I_2$ there and reallocate them into the other eight bins simply by allocating them where a task from $I_1$ was allocated. By repeating this procedure we can reallocate the tasks from $I_2$ into at most $\left(\frac{8}{9} + 2\varepsilon\right)OPT$ bins, proving the claim. Finally, if a bin gets 6 or less tasks, we can always sort to get a non-increasing load profile. Otherwise, if NFDH does not pack the tasks, then their area is at least $\frac{1}{2}TC$ as the three tasks defining the shelves have total height larger than $C$ and the remaining at least four tasks will also have total height larger than $C$. This proves the claim for this case.
	
				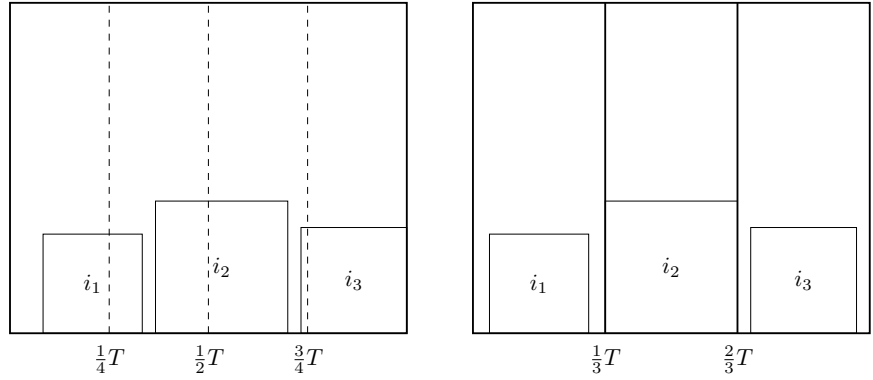
\begin{figure}
	    \centering
	    \resizebox{0.7\textwidth}{!}{
		\begin{tikzpicture}
		  \draw[thick] (1,0) rectangle (7,5);

  		     		  \draw[dashed] (2.5, 0) -- (2.5, 5);
  		              \draw[dashed] (4, 0) -- (4, 5);
  		     		  \draw[dashed] (5.5, 0) -- (5.5, 5);
                     \draw (2.5, -0.4) node {$\frac{1}{4}T$};
                     \draw (4, -0.4) node {$\frac{1}{2}T$};
                     \draw (5.5, -0.4) node {$\frac{3}{4}T$};

		  \draw (1.5,0) rectangle (3,1.5);
		  \draw (3.2,0) rectangle (5.2,2);
		  \draw (5.4,0) rectangle (7,1.6);
                      \draw (2.25, 0.75) node {$i_1$};
                     \draw (4.2, 1) node {$i_2$};
                     \draw (6.2, 0.8) node {$i_3$};

  		   \draw[thick](10, 0) -- (10, 5);
  		              \draw[thick] (12, 0) -- (12, 5);
                     \draw (10, -0.4) node {$\frac{1}{3}T$};
                     \draw (12, -0.4) node {$\frac{2}{3}T$};
                     
        \draw (8.25,0) rectangle (9.75,1.5);
		  \draw (10,0) rectangle (12,2);
		  \draw (12.2,0) rectangle (13.8,1.6);
                     \draw (9, 0.75) node {$i_1$};
                     \draw (11, 1) node {$i_2$};
                     \draw (13, 0.8) node {$i_3$};

\draw[thick] (8,0) rectangle (14,5);
		\end{tikzpicture}}
		\caption{Depiction of how square tasks of width larger than $\frac{1}{4}T$ and at most $\frac{1}{3}T$ can be rearranged into three disjoint intervals of slots having total width $\frac{1}{3}T$. This way each interval induces a one-dimensional bin.}
		\label{fig:2eps3reg}
			\end{figure} 
			
\item The case $\frac{3}{4}C < T \le C$ can be addressed in a similar way to the previous case so as to allocate tasks having height larger than $\frac{1}{4}C$. This way, the remaining tasks to be allocated will have height at most $\frac{1}{4}C$ and width at most $\frac{1}{4}T \le \frac{1}{3}C$ and it holds that $\left(1-\frac{1}{4}\right)\left(1-\frac{1}{3}\right) = \frac{1}{2}$. If $OPT \le \frac{1}{\varepsilon}$, we guess the optimal packing of big tasks by brute force, and rearrange each bin that is not $\frac{1}{2}$-full using NFDH into at most two bins. This is possible since if a task $i_{k+1}$ does not fit, then the total area of tasks placed is at least \[h_{i_2}(T - w_{i_2}) + h_{i_3}(T - w_{i_3})+...+h_{i_{k+1}}(T - w_{i_{k+1}}) \ge \left( \sum_{n \ge 2: w_{i_n} \le \frac{1}{2}T} h_{i_n
} \right) \frac{1}{2}T \ge  \frac{1}{2}TC,\] where the tasks $i_1, i_2, \dots, i_{k+1}$ are the tasks defining the shelves. The only difference with the previous analysis is that now there can be more than one task of width larger than $\frac{1}{2}T$, but their total height cannot be larger than $C$ as they fit together in the bin. If on the other hand $OPT>\frac{1}{\varepsilon}$, we can again allocate optimally the tasks of height larger than $\frac{1}{3}C$ and the remaining tasks with the help of the APTAS (this can be done analogously as before by rotating the bins by $90$ degrees). This time it is easier to obtain non-increasing load profiles for the bins which are not $\frac{1}{2}$-full as they are naturally partitioned into shelves.

\item Finally, the case $T< \frac{3}{4}C$ can be addressed similarly as in the first case. Our goal will be to allocate all the tasks of width larger than $\frac{1}{3}T$ in an initial phase, as then the remaining tasks would have width at most $\frac{1}{3}T$ and height at most $\frac{1}{3}T \le \frac{1}{4}C$ and it holds that $\left(1-\frac{1}{3}\right)\left(1-\frac{1}{4}\right) = \frac{1}{2}$. We can apply the same construction as in the first case but this time to allocate all the tasks except for one into a region of width $T$ and height $OPT\cdot C$. This packing is decomposed into two piles, one to the left where tasks are packed in non-increasing order of width starting from the bottom, and one to the right which is packed in non-increasing order of width starting from the top. Furthermore, each task in the left pile is larger in width and height compared to any task in the left pile. We can then draw horizontal lines at each possible multiple of $C$ to induce rectangular regions of height $C$ and width $T$.

Analogously to the first case, we will consider all the obtained regions from bottom to top and rearrange the set of tasks intersecting the region into two bins as follows: We allocate into one bin all the tasks from the left pile completely contained in the region plus the task from the right pile intersecting the lower boundary of the region (if any), and allocate into another bin all the tasks from the right pile completely contained in the region plus the task from the left pile intersecting the upper boundary (if any, including the task which is not packed as if it were intersecting the topmost boundary). This way we obtain a solution for all the tasks that uses at most $2OPT$ bins and we will show how to obtain non-increasing load profiles. Bins containing the tasks from the left pile completely contained plus the task from the right pile intersecting the lower boundary can be rearranged by simply allocating the task from the right pile next to any task from the left pile which as if they were in the same shelf (this does not induce conflicts as tasks in the right pile are smaller in height than any task in the left pile). For each of the remaining bins we will rearrange the tasks inside using NFDH. Notice that the width of the task from the left pile plus the width of any task from the right pile in the bin is at most $T$, meaning that each shelf will have exactly two tasks, that the width of the task in the right pile is at most $\frac{2}{3}T$ and consequently that the total height of tasks in the bin is at most $C+\frac{2}{3}T \le \frac{3}{2}C$. If NFDH does not allocate all the tasks in the bin, then the total height of tasks allocated is at least $h_{i_1} + 2(h_{i_2} + h_{i_3} + \dots + h_{i_{k+1}})$, where $i_1, \dots, i_k$ are the tasks defining the shelves and $i_{k+1}$ is the task that did not fit. We conclude by noticing that \[h_{i_2} + h_{i_3} + \dots + h_{i_{k+1}} > C-h_{i_1},\] implying that the total height of the tasks in the bin is strictly larger than $2C - h_{i_1} \ge 2C - \frac{2}{3}T \ge \frac{3}{2}C$ which is a contradiction. This concludes the proof.\qedhere
\end{enumerate}
\end{proof}

\begin{theorem}\label{thm:2apx-square-gen} There is a $2$-approximation for the 2D Demand Bin Packing problem for Square Tasks. \end{theorem}

\begin{proof}
    Thanks to Lemma~\ref{lem:2epslem}, there always exist parameters $\alpha,\beta\ge 0$ such that a $2$-structured solution for tasks having height larger than $\alpha C$ or width larger than $\beta T$ can be computed, satisfying that $(1-\alpha)(1-\beta) \ge \frac{1}{2}$. The remaining tasks can be placed on top by means of Lemma~\ref{lem:smalltasks} using at most $2OPT$ bins, obtaining the desired solution. \end{proof}

\end{document}